\documentclass[lettersize,journal]{IEEEtran}
\usepackage{amsmath,amsfonts}
\usepackage{algorithmic}
\usepackage{algorithm}
\usepackage{array}
\usepackage{textcomp}
\usepackage{stfloats}
\usepackage{url}
\usepackage{verbatim}
\usepackage{graphicx}
\usepackage{cite}
\usepackage{amsthm,amsmath,amssymb}
\usepackage{mathrsfs}
\usepackage{enumitem}
\usepackage{subcaption} 
\usepackage{times}
\usepackage{svg}
\usepackage{multirow}
\usepackage{amsthm}            
\usepackage{longtable}
\usepackage{booktabs}
\usepackage{dsfont}

\usepackage{tabularx}
\usepackage{booktabs}

\theoremstyle{plain}

\newtheorem{definition}{Definition}

\newtheorem{proposition}{Proposition}
\newtheorem{problem}{Problem}

\begin{document}

\title{THz RHS Transceiver for Low-Latency Multi-User VR Transmission with MEC}

\author{Liangshun Wu, Wen Chen, Honghao Wang, Zhendong Li, Ying Wang
\thanks{Liangshun Wu, Wen Chen and Honghao Wang are with the Department of Electronic Engineering, Shanghai Jiao Tong University, Shanghai, China (e-mail: wuliangshun@sjtu.edu.cn; wenchen@sjtu.edu.cn, mc25018@um.edu.mo). Zhendong Li is with the School of Information and Communication Engineering, Xi’an Jiaotong University, Xi’an, China (e-mail: lizhendong@xjtu.edu.cn).  Ying Wang is with the State Key Laboratory of Networking and Switching Technology, Beijing University of Posts and Telecommunications, Beijing, China (e-mail: wangying@bupt.edu.cn). }
}


\maketitle

\begin{abstract}
This paper investigates a Terahertz (THz)-enabled mobile edge computing (MEC)-assisted virtual reality (VR) system using reconfigurable holographic surfaces (RHS) as transceiver for multi-user beamforming and holographic-pattern division multiple access (HDMA). We develop an end-to-end model  for the 3D field-of-view (FoV) generation pipeline and optimize content prefetching, rendering offloading  under memory and power constraints, and beamforming  accommodating user movement by adjusting holographic pattern weights for beam shaping and feeds power allocation for  excitation amplitude adjustment. For homogeneous FoVs, we derive closed-form policies for prefetching 2D or 3D FoVs or direct transmission of 3D FoVs. For heterogeneous FoVs, we exploit the timescale separation between prefetching/rendering and fast RHS beamforming, decomposing the optimization into a rendering-prefetching combinatorial optimization problem and a short-timescale beamforming convex optimization problem. Simulations show significant latency reductions under tight resource constraints.
\end{abstract}


\section{Introduction}
\IEEEPARstart{I}{m}mersive $360^\circ$ virtual reality (VR) is envisioned as a key 6G application for lifelike experiences in entertainment, education, and beyond \cite{huang2023rate}. Yet, wireless VR faces three critical challenges: (i) huge data rates, as 4K $360^\circ$ video demands 70-300Mbps \cite{9606825}, far beyond sub-6 GHz\cite{zhinuk2024spectral}; (ii) ultra-low latency, with motion-to-photon delay needing to be under 20ms (i.e., 0.02s) \cite{picano2021end,zhao2024mobility,tun2024joint}, leaving only milliseconds for transmission after rendering/decoding; and (iii) limited head-mounted display (HMD) resources \cite{9120235}. These bottlenecks render high-quality wireless VR impractical on today’s networks. 

To address this, Terahertz (THz) communication is emerging enabler for offering multi-Gbps links (0.1–10 THz)\cite{akyildiz2022terahertz,jiang2024terahertz}, which is suitable for real-time panoramic streaming, 6G VR and holographic applications\cite{sun2019communications,9565222,9745789,9678373}.

\begin{figure*}[h]
\centering
\includegraphics[width=\textwidth]{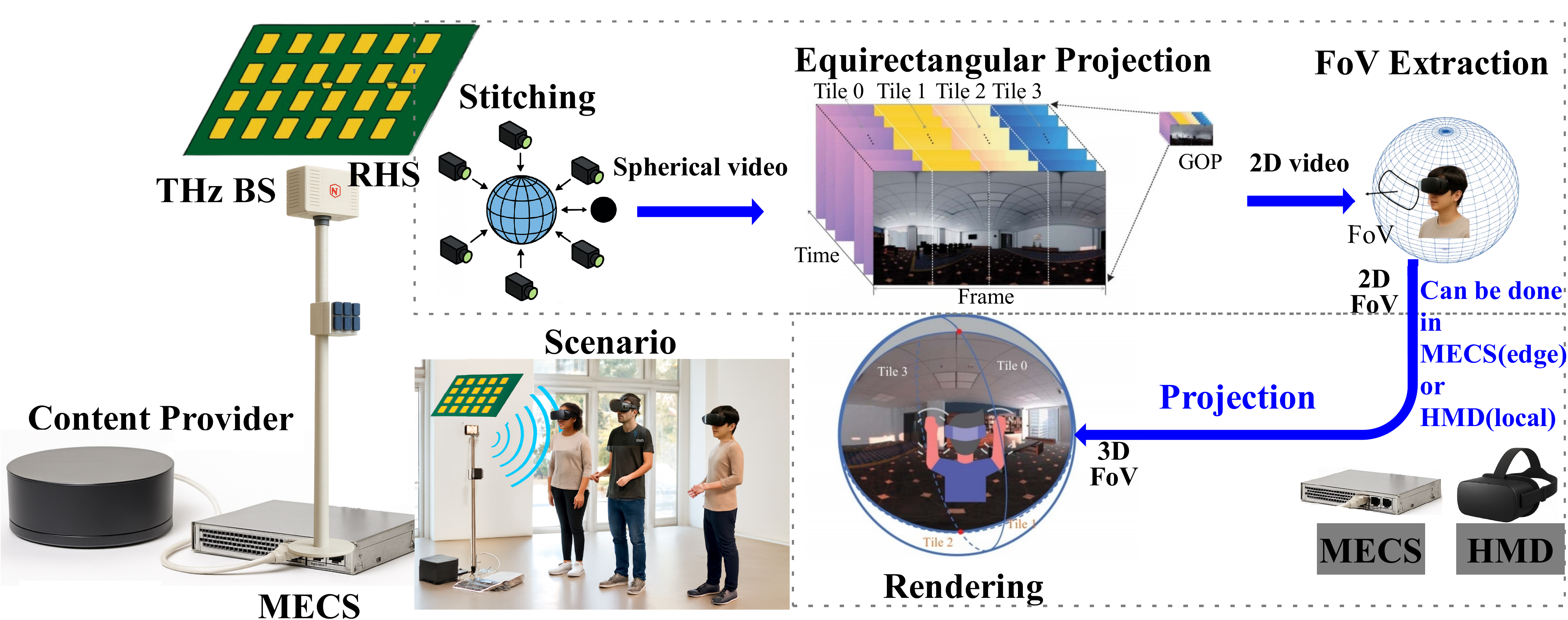}
\caption{A typical \(360^\circ\) VR video production pipline using THz RHS transceiver for wireless transmission with MEC.}
\label{fig:prod}
\end{figure*}

\subsection{VR Transmission with MEC}
We first analyze a typical \(360^\circ\) VR video production pipeline (Fig.~\ref{fig:prod}) \cite{sun2019communications}: i) \textit{Stitching}: Videos captured by a multi-camera array are stitched together to generate a spherical video; ii) \textit{Equirectangular projection}:  the spherical video is unfolded to 2D video encoded in group of pictures (GOPs) consisting of I-, P-, and B-frames. For high-resolution frames, each frame is split into multiple chunks, and each chunk is further divided into  \(4\times 6 = 24\) tiles. For low-resolution frames, the frame is directly partitioned into 24 tiles. iii) \textit{Extraction}: the current FoV is determined based on the user's gaze direction, and corresponding 2D FoV is extracted from the 2D video according to the FoV tracked by the HMD; iv) \textit{Projection}: The 2D FoV is projected into a 3D FoV; v) \textit{Rendering}: The resulting 3D FoV is rendered on the HMD.
The components i) \textit{stitching}, ii) \textit{equirectangular projection}, and iii) \textit{extraction} are performed at the cloud-based content provider, since they require the full \(360^\circ\) video as input. Offloading these tasks reduces the rendering burden on the mobile edge computing server (MECS) and the HMD, and alleviates wireless traffic. The content provider is connected to the MECS and the Terahertz (THz) base station (BS) via high-capacity optical fiber links. The BS delivers 2D FoVs to multiple HMD-equipped users, while the HMDs execute v) \textit{projection} and v) \textit{rendering} to display the immersive content.

The system leverages mobile edge computing (MEC) \cite{wu2025towards,wu2024attention} to offload rendering and decoding tasks from head-mounted displays (HMDs), thereby reducing device energy consumption and latency \cite{9350227,9667509}. To further improve efficiency, both 2D and 3D Fields of View (FoVs) are prefetched at the MEC server (MECS), with 3D FoVs generated offline by a low-complexity projection module. Since users’ FoV requests exhibit strong popularity patterns, probabilistic models can be exploited to optimize prefetching decisions \cite{zheng2021stc}. With the increasing computational capability of HMDs, the projection module can be migrated to the HMD, enabling local 3D FoV generation from prefetched 2D FoVs and reducing wireless traffic by up to half, as stereoscopic rendering requires duplicating 3D FoVs for both eyes \cite{deng2022fov}. However, this migration introduces additional rendering delay, necessitating careful prefetching and computation strategies. The HMD memory can store either 2D or 3D FoVs, where 3D prefetching lowers latency and energy consumption at the cost of higher memory usage, making the tradeoff between communication, computation, and storage a key design challenge. Most existing MEC-based VR studies oversimplify rendering and lack rigorous pipeline-level modeling \cite{9667509,10816367,9350227}, while HMD-side prefetching approaches often neglect essential VR characteristics or focus solely on 2D FoVs \cite{ma2024qoe,wang2022meta}. In contrast, this work models the full pipeline and jointly optimizes 2D and 3D FoV prefetching and rendering.

In the system depicted in Fig.~\ref{fig:prod}, we adopt Reconfigurable Holographic Surfaces (RHS) as the BS transceiver for precise beamforming, benefiting for efficient 2D and 3D FoV transmission from the MECS to multiple HMDs, ensuring optimal coverage when users move. THz technology facilitates the high-speed transmission of 2D FoVs from the MEC to HMDs, enabling low-latency, high-bandwidth communication essential for immersive VR experiences.

\subsection{RHS}
RHS use sub-wavelength units with carefully designed resonant properties to enable stable broadband holographic modulation across hundreds of GHz without needing separate structures for different frequencies. RHS controls the wavefront via spatial phase distribution, ensuring consistent wavefront control across a wide bandwidth\cite{9696209}. Using electrically controlled units (e.g., PIN diodes, varactors), holographic patterns can switch in microseconds, enabling real-time beam steering and multi-user shaping\cite{9994740}, ideal for dynamic immersive VR. With ultra-thin PCB manufacturing and no need for independent phase shifter (PA), hardware costs and power consumption are low, facilitating large-scale arrays with hundreds/thousands of elements. Additionally, holographic-pattern division multiple access (HDMA)\cite{9681843,yutong2022holographic} utilizes the superposition of holographic patterns associated with different beams to serve multiple HMDs. Current RHS studies mainly address hardware and beamforming, while their integration with MEC and THz VR transmission remains unexplored.

\subsection{Contributions}

The main contributions of this paper are as follows:
\begin{enumerate}
    \item {THz and RHS-assisted VR transmission architecture:}  
    We propose a VR transmission architecture based on RHS (holographic metasurface) and THz (terahertz) technologies, supporting multi-user systems through precise beam control and electronically steerable multi-beams. This scheme leverages THz's ultra-high bandwidth and RHS's fast beamforming to accommodate multi-user mobility in immersive VR environments, enabling efficient VR data transmission. We also incorporate mutual coupling effects, which account for the interactions between adjacent radiation elements in RHS array.
    
    \item {MEC-based VR framework and formulation:}  
    We design an MEC-based VR framework that offloads rendering tasks to the MEC-HMD edge, reducing wireless traffic and latency. We formulate and solve a joint prefetching-rendering-beamforming optimization problem to minimize transmission delay under constraints, optimizing content prefetching, rendering offloading, and re-beamforming with holographic pattern weight adjustments and feed power allocation to adapt to user mobility.
    
    \item {Binary-continuous mixed optimization for joint prefetching-rendering optimization with mobility-aware RHS beamforming:}  
    This involves a challenging joint binary-continuous optimization problem. We first derive closed-form solutions for prefetching 2D FoVs with local rendering, prefetching 3D FoVs, or directly transmitting remote 3D FoVs. We extend the framework to heterogeneous FoV settings, considering HMD mobility and rapid beam re-steering. By decomposing the problem into long- and short-timescale subproblems, we combine Multiple-Choice Multi-Dimensional Knapsack Problem (MMKP)-based heuristics, binary-relaxation Difference-of-Convex Algorithm (DCA) and Convex-Concave Procedure (CCCP), full DCA/CCCP solvers, and Projected Gradient (PG)-based RHS beamforming optimizers to efficiently solve the optimization problem.
\end{enumerate}

Next, Section II describes the system model; Sections III and IV study the homogeneous and heterogeneous cases; Section V presents simulation results; and Section VI concludes.

\begin{figure*}[h]
\centering
\includegraphics[width=\textwidth]{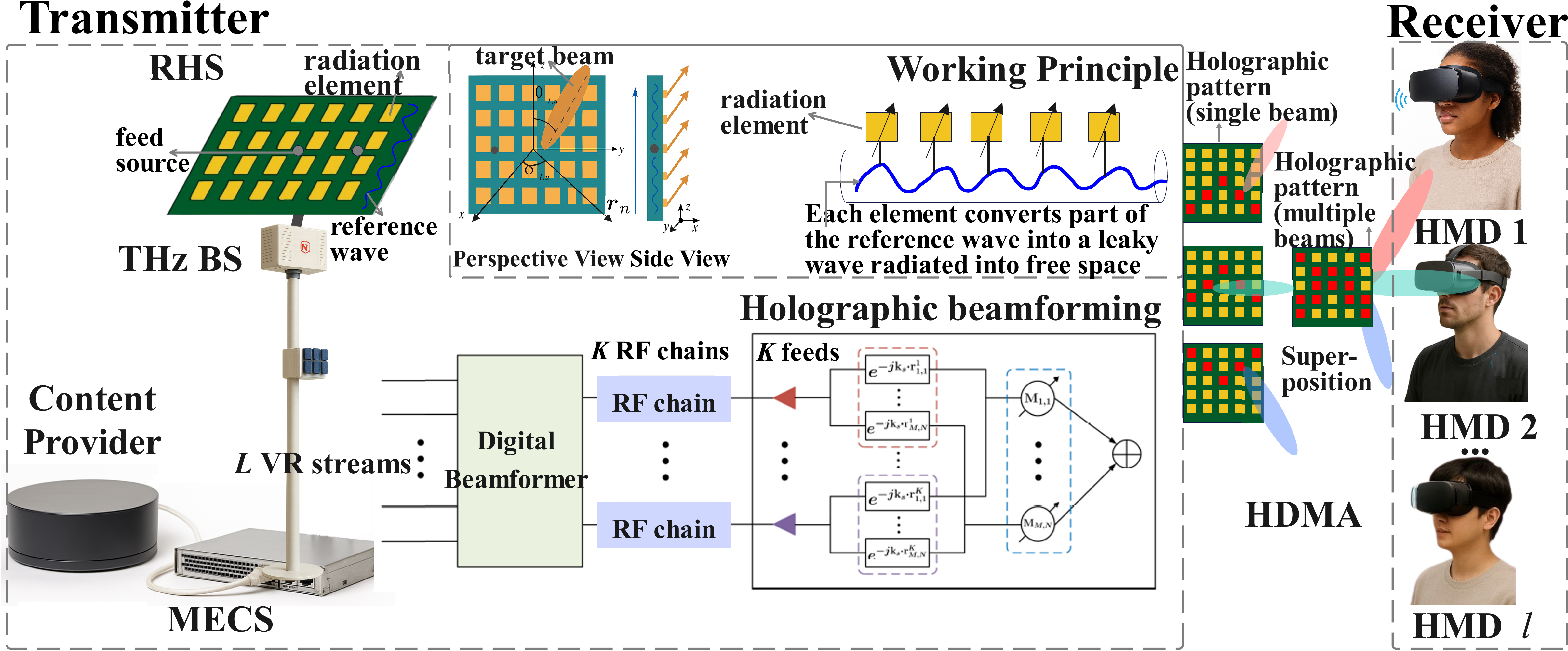}
\caption{RHS-assisted THz multi-user VR transmission.}
\label{scenario}
\end{figure*}
\section{System Model}
\subsection{THz RHS Channel Model }
Rendering offloading involves communication between MECS, equipped with a BS and an RHS to support highly directional and adaptive THz beamforming, and $L$ HMDs.

\subsubsection{Frequency-Selective THz Channel}
Due to the extremely high carrier frequency, the THz channel exhibits strong frequency 
selectivity, molecular absorption, and severe spreading loss. Following common practice,  the carrier frequency is $f_c=0.3~\mathrm{THz} = 300~\mathrm{GHz} $, hence $\lambda_c=c/f_c=1~\mathrm{mm}$, where $c = 3\times 10^8~\mathrm{m/s}$ is the speed of light.
The entire bandwidth $B_c=40~\mathrm{GHz}$ is divided into $U=4$ subbands, each with $B_g=10~\mathrm{GHz}$, which is smaller than the channel coherence bandwidth. Thus, each 
subband can be modeled as a narrowband flat-fading channel. Let $f_u$ denote the center frequency of the $u$-th subband, $f_u= f_c-\frac{B_c}{2}+\left(u-\frac{1}{2}\right)\frac{B_c}{U}$, i.e., $f_u\in\{285,295,305,315\}~\mathrm{GHz}$.  Under the subband narrowband assumption, the frequency-selective THz baseband channel coefficient between the BS and the $l$-th HMD ($l \in \mathcal{L}, |\mathcal{L}|=L$) on subband $u$ is expressed as
$
h_{l,u}
= \frac{c}{4\pi f_u d_l}\exp\left(-\frac{\kappa(f_u) d_l}{2}\right),
\label{eq:h}
$
where $\kappa(f_u)$ denotes the frequency-dependent molecular absorption coefficient~\cite{du2025qoe}. $d_l$ denotes the Euclidean distance between the RHS and HMD $l$:
$
d_l = \|\mathbf{p}_l - \mathbf{p}_{\mathrm{R}}\|
= \sqrt{\Delta x_l^2+\Delta y_l^2+\Delta z_l^2},
$ where $\mathbf{p}_{\mathrm{R}} = [x_{\mathrm{R}},y_{\mathrm{R}},z_{\mathrm{R}}]^{\mathrm{T}}$ is the RHS position and 
$
\mathbf{p}_l = [x_l,y_l,z_l]^{\mathrm{T}}.
$ is HMD  $l$'s position. 
$\Delta x_l=x_l-x_{\mathrm{R}},$$\Delta y_l=y_l-y_{\mathrm{R}},$$\Delta z_l=z_l-z_{\mathrm{R}}.$

\subsubsection{Inter-Band Interference (IBI)}
Although each subband is narrowband, spectral leakage across adjacent subbands causes 
IBI. According to~\cite{du2025qoe}, the IBI on subband $u$ 
can be approximated as a zero-mean Gaussian random variable with variance
\begin{equation}
I_u \sim \mathcal{N}\left( 0,\; 
\int_{f_u} \sum_{v \neq u}^{U} P_v 
\left| G_v(f_u) 
\sum_{q \in \mathcal{Q}} \alpha_{v,q}
\right|^2 df_u
\right),
\label{eq:IBI}
\end{equation}
where  $\alpha_{v,q}$ is leakage coefficient, $P_v$ is the power allocated to subband $v$ and $G_v(\cdot)$ is the waveform 
spectrum.

\subsubsection{HDMA-Based RHS Holographic Beamforming}
RHS is a special leaky-wave antenna composed of feed sources and a large number of metasurface radiation elements. As shown in Fig.~\ref{scenario}, each feed is placed at the edge or backplane of the RHS and converts the input signal into an electromagnetic reference wave. The RHS adopts a serial-fed architecture: the reference wave propagates across the surface and sequentially excites the radiation elements. Each element converts part of the reference wave into a leaky wave radiated into free space, where the radiation amplitude at each element can be independently tuned to realize holographic beamforming. Compared with conventional parallel-fed phased arrays, the serial-fed RHS architecture greatly simplifies wiring for large-scale MIMO deployment.

To illustrate the holographic principle, consider an RHS with $K$ feed sources and $N=N_x\times N_y$ radiation elements, $N_x = N_y=\sqrt{N}$, arranged as a $\sqrt{N}\times\sqrt{N}$ uniform planar array (UPA) with inter-element spacing $d=\lambda_c/2$. Let $\Psi_{l,u}$ denote the desired beam radiating toward HMD $l$ direction  $(\theta_{l,u},\varphi_{l,u})$ for subband $u$  and $\Psi_{\mathrm{ref}}$ the reference wave generated by feed $k$. At the $n$-th radiation element, they can be written as
$
\Psi_{l,u}=e^{-j\boldsymbol{k}_f(\theta_{l,u},\varphi_{l,u})\cdot \boldsymbol{r}_n},
$
$
\Psi^{\mathrm{ref}}(\boldsymbol{r}_n^k)=e^{-j\boldsymbol{k}_s\cdot \boldsymbol{r}_n^k},
$
where $\boldsymbol{k}_f(\theta_{l,u},\varphi_{l,u})$ is the free-space wave vector, $\boldsymbol{r}_n$ is the position of the $n$-th element, $\boldsymbol{k}_s$ is the propagation vector of the reference wave, and $\boldsymbol{r}_n^k$ is the displacement vector from feed $k$ to element $n$. According to holographic interference, the recorded hologram is the interference pattern
$
\Psi_{l,u}^{\text{intf}}=\Psi_{l,u}\Psi^{\text{ref}}.
$
When illuminated again by the reference wave, the resulting leaky-wave field satisfies
$
\Psi_{l,u}^{\text{intf}}\Psi^{\text{ref}}\propto \Psi_{l,u}\left|\Psi^{\text{ref}}\right|^2.
$
The departure direction of the LoS path from the RHS to HMD $l$ can be parameterized in spherical coordinates by the elevation angle $\theta_{l,u}$ and azimuth angle $\varphi_{l,u}$ as $\varphi_{l,u}= \arctan2\big(\Delta y_l,\Delta x_l\big)$, $\theta_{l,u}= \arctan2\Big(\Delta z_l,\sqrt{\Delta x_l^2+\Delta y_l^2}\Big)$, which are independent of the subband $u$ and fully determined by the geometry.

Based on the holographic interference principle, the composite holographic pattern of subband $u$ used in HDMA for all $L$ HMDs is obtained by superimposing the corresponding interference fields:
$
m_{u,n}
=\sum_{l=1}^{L}\sum_{k=1}^{K}
a_{l,u,k}
\frac{\operatorname{Re}\left[\Psi_{l,u}^{\text{intf}}\right]+1}{2},
$
where $a_{l,u,k}\ge 0$ is the amplitude ratio for the beam 
pointing to user $l$ from feed $k$, the weights satisfying $\sum_{l=1}^{L}\sum_{k=1}^{K} a_{l,u,k}=1$. The $n$-th radiation element from $k$ feed of holographic beamforming matrix $\boldsymbol{M}_{u}\in\mathbb{C}^{N\times K}$ is defined as
$
{\boldsymbol{M}_{u}}_n^k
=\sqrt{\beta} m_{u,n}
e^{-\alpha \|\boldsymbol{r}_n^k\|}
e^{-j \boldsymbol{k}_s \boldsymbol{r}_n^{k}},
$
where $\beta$ denotes element radiation efficiency, $\alpha$ is the structural attenuation coefficient.

\begin{figure}[h]
\centering
\includegraphics[width=0.48\textwidth]{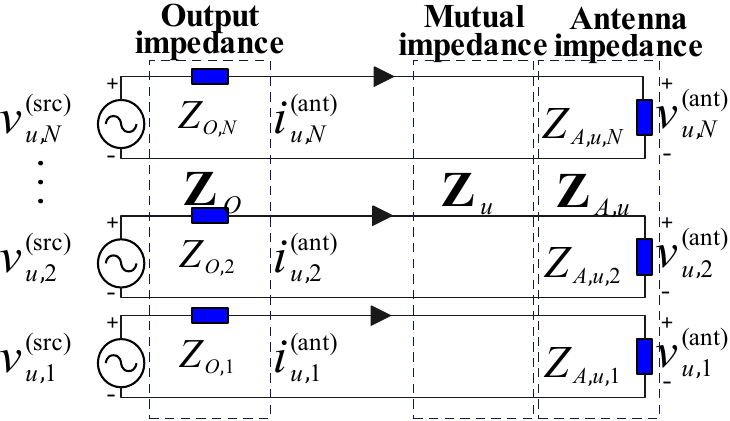}
\caption{An equivalent multi-port circuit model of the mutual coupling effect at the transmitting end (Tx)  in RHS.}
\label{fig:array}
\end{figure}

Due to the closer spacing between radiation elements, mutual coupling occurs as discussed in \cite{11274834,11274829}. The equivalent multi-port circuit model of the RHS at the transmitter is illustrated in Fig. \ref{fig:array} \cite{11274829}. Each RHS element is excited by a voltage source with output impedance $Z_O$, and the resulting antenna currents flow through a frequency-dependent mutual impedance network $\boldsymbol{Z}_u$, whose diagonal and off-diagonal entries represent the self-impedance and inter-element coupling, respectively. The antenna terminal voltages $\mathbf{v}_u^{(\mathrm{ant})}$ are therefore jointly determined by the source excitations, the output impedance, and the mutual impedance network. The voltages on the antenna elements, denoted as $\mathbf{v}_u^{(\mathrm{ant})}$, is given by $\mathbf{v}_u^{(\mathrm {ant })}=\mathbf{Z}_u \mathbf{i}_u^{(\mathrm {ant })}$,
where $\mathbf{Z}_u$ denotes the frequency-dependent mutual impedance matrix at frequency $f_u$. 
The mutual impedance matrix $\mathbf{Z}_u\in\mathbb{C}^{N\times N}$ is constructed element-wise as
$[\mathbf{Z}_u]_{p,q}=Z_{A,u}$ for $p=q$, while for $p\neq q$ it is given by the mutual impedance
between two electrically short dipoles separated by distance $d_{p,q}$, expressed as
$
[\mathbf{Z}_u]_{p,q}
=
-30\{
\int_{0}^{l_d}
\sin \left(\frac{2\pi}{\lambda_u}z\right)+\int_{l_d}^{2 l_d} \sin \left[\frac{2 \pi}{\lambda_u}\left(2 l_d-z\right)\right]\}
$
$
(
\frac{-\mathrm{j}e^{-\mathrm{j}\frac{2\pi}{\lambda_u}r_1}}{r_1}
+
\frac{-\mathrm{j}e^{-\mathrm{j}\frac{2\pi}{\lambda_u}r_2}}{r_2}
+
\frac{2\mathrm{j}\cos\!\left(\frac{2\pi}{\lambda_u}l_d\right)
e^{-\mathrm{j}\frac{2\pi}{\lambda_u}r_0}}{r_0}
)\mathrm{d}z,
$
where $r_0=\sqrt{d_{p,q}^2+z^2}$,
$r_1=\sqrt{d_{p,q}^2+(l_d-z)^2}$,
$r_2=\sqrt{d_{p,q}^2+(l_d+z)^2}$,
and the inter-element distance is
$d_{p,q}=d\sqrt{(i_p-i_q)^2+(j_p-j_q)^2}$ under the $N=N_x\times N_y$ planar array indexing; $l_d=\lambda_c/4$ is the length of the electrical short dipole antenna. Here, $z$ is the variable that represents the position along the length of the antenna, used to calculate the distances at different positions along the antenna.

According to Kirchhoff's laws, the voltages generated by the sources in the presence of mutual coupling is $\mathbf{v}_u^{(\mathrm{src})}=\mathbf{Z}_O \mathbf{Z}_u^{-1} \mathbf{v}_u^{(\mathrm{ant})}+\mathbf{v}_u^{(\mathrm{ant})}$, where $\mathbf{Z}_O = Z_O \mathbf{I}_N$ denotes the output impedance matrix of the voltage sources, typically $Z_O = 50 \Omega$ for RF power amplifiers. In the absence of mutual coupling, the mutual impedance matrix degenerates to the antenna impedance matrix, i.e., $\mathbf{Z}_u \xrightarrow{\text { No Mutual Coupling }} \mathbf{Z}_{A, u}=Z_{A, u} \mathbf{I}_N$. Since RHS adopts electrically short dipoles ($l_d\ll\lambda_u$), $Z_{A,u}$ follows the standard approximation in \cite{11274834}: $Z_{A, u} \approx 80\left(\frac{l_d}{\lambda_u}\right)^2-\mathrm{j} \frac{120}{\pi} \ln \left(\frac{\lambda_u}{2 \pi l_d}\right) $. the voltages generated by the sources
in the absence of mutual coupling can be obtained as:  $\mathbf{v}_u^{(\mathrm{src})}=\mathbf{Z}_O \mathbf{Z}_{A, u}^{-1} \mathbf{v}_u^{(\mathrm{ant}, \mathrm{no} \mathrm{mc})}+\mathbf{v}_u^{(\mathrm{ant}, \mathrm{nomc})}$, where $\mathbf{v}_u^{(\mathrm {ant, nomc })}$ encapsulates the voltages on the antenna elements in the absence of mutual coupling. So, We have $\mathbf{Z}_O \mathbf{Z}_{A, u}^{-1} \mathbf{v}_u^{(\mathrm {ant,no mc })}+\mathbf{v}_u^{(\mathrm {ant,no mc })}=\mathbf{Z}_O \mathbf{Z}_u^{-1} \mathbf{v}_u^{(\mathrm {ant })}+\mathbf{v}_u^{(\mathrm {ant })}$, that means $\left(\mathbf{I}_N+\mathbf{Z}_O \mathbf{Z}_u^{-1}\right)^{-1}\left(\mathbf{I}_N+\mathbf{Z}_O \mathbf{Z}_{A, u}^{-1}\right) \mathbf{v}_u^{(\mathrm {ant,nomc })}=\mathbf{v}_u^{(\mathrm {ant })}$. We define the coupling matrix $\boldsymbol{\Xi}_u \in \mathbb{C}^{N\times N}$ as
$
\boldsymbol{\Xi}_u=\left(\mathbf{I}_N+\mathbf{Z}_O \mathbf{Z}_u^{-1}\right)^{-1}\left(\mathbf{I}_N+\mathbf{Z}_O \mathbf{Z}_{A, u}^{-1}\right)
$.
The holographic beamforming matrix becomes
\begin{equation}
\tilde{\boldsymbol{M}}_u
\triangleq
\boldsymbol{\Xi}_u  \boldsymbol{M}_u .
\end{equation}

Let $\boldsymbol{P}_u\in\mathbb{C}^{K\times 1}$ be the power allocation vector for subband $u$ with $\mathrm{Tr}(\boldsymbol{P}_u\boldsymbol{P}_u^H)=P_{u}^\mathrm{Tx}$, where $P_{u}^\mathrm{Tx}=\frac{P^\mathrm{Tx}}{U}$ is the total power of the transmitter for subband $u$ (equally shared by all $U$ subbands), and $x_{l,u}$ the broadcast signal. The received signal of HMD $l$ of subband $u$  is
$
y_{l,u}=\boldsymbol{H}_{l,u}\tilde{\boldsymbol{M}}_u\boldsymbol{P}_u x_{l,u}+n_l,
$
where  $n_l$ is Gaussian noise with variance $B_gn_0$ and  $\boldsymbol{H}_{l,u}\in\mathbb{C}^{1\times N}$ is the RHS-HMD channel vector.
Assuming that all elements are approximately equidistant from the HMD $l$ (i.e., the spacing among the RHS elements is negligible compared with the link distance), the channel coefficients of all elements in $\boldsymbol{H}_{l,u}$ become identical. Thus, the channel gain of the  $n$-th radiation element is 
$
h_{l,u,n} = h_{l,u}.
$

The $u$-th subband suffers from IBI $I_u$ caused by spectral leakage from other subbands, as modeled in eq.~\eqref{eq:IBI}. The instantaneous SINR of HMD 
$l$ on subband $u$ is given by
\begin{equation}
\gamma_{l,u}
=
\frac{
G^{\mathrm{Tx}} G^{\mathrm{Rx}}
\big|  \boldsymbol{H}_{l,u} \tilde{\boldsymbol{M}}_u\boldsymbol{P}_u  \big|^2 
}{
G^{\mathrm{Tx}} G^{\mathrm{Rx}}I_u + B_g n_0
},
\label{eq:SINR_RHS_HDMA}
\end{equation}
where $G^{\mathrm{Tx}}$ and $G^{\mathrm{Rx}}$ denote the transmit and receive antenna 
gains, respectively.

Given a target bit error rate (BER) $\epsilon_{l,u}$ for HMD $l$ on subband $u$, the 
achievable spectral efficiency with $M$-QAM is approximated by $k_{l,u}$ (bit/s/Hz): $k_{l,u} = \log_2\left(
1 - \frac{1.5\gamma_{l,u}}{\ln\big(5\epsilon_{l,u}\big)}
\right)$.

HMD $l$'s individual transmission rate for all $U$ subbands is
\begin{equation}
R_l
= B_g \sum_{u=1}^{U} k_{l,u} = B_g \sum_{u=1}^{U} \log_2\left(
1 - \frac{1.5\gamma_{l,u}}{\ln\big(5\epsilon_{l,u}\big)}
\right).
\label{eq:user_rate_HDMA}
\end{equation}

\subsection{VR Transmission Model with MEC}
Let $\mathcal{V} \triangleq \{1, \cdots, V\}$ denote the FoV space, where these FoVs are extracted from 2D videos, as shown in Fig. \ref{scenario}. For HMD $l$, the projection from each FoV $i \in \mathcal{V}$ from 2D FoV to 3D FoV is represented by the triplet $(Q_l^{2\mathrm{D}}, Q_l^{3\mathrm{D}}, o)$, where $Q_l^{2\mathrm{D}}$ and $Q_l^{3\mathrm{D}}$ are the data sizes (in bits) of the 2D FoV and 3D FoV that HMD $l$ needs to process, respectively, and $o$ is the number of rendering cycles required per input bit (unit: cycles/bit). Let $\alpha_l \triangleq \frac{Q_l^{3\mathrm{D}}}{Q_l^{2\mathrm{D}}}$ denote the ratio between the 3D FoV size and the 2D FoV size processed by HMD $l$. Typically, $\alpha_l \geq 2$, which is required for creating binocular stereoscopic visual effects.

Similar to \cite{liu2023rendered}, the request stream of HMD $l$ follows the Independent Reference Model (IRM). It is assumed that the probability that HMD $l$ requests FoV $i \in \mathcal{V}$ at each request is a constant $\pi_{l,i}$, and the requests are independent of each other. The probability $\pi_{l,i}$ represents the frequency with which HMD $l$ selects FoV $i$ when browsing the scene, and satisfies $\sum_{i=1}^V \pi_{l,i} = 1$.  Moreover, to avoid motion sickness and nausea, each request of HMD $l$ must be served within the maximum allowable delay $\tau$.

First, consider the prefetching layout of HMD $l$. Assume that HMD $l$ has a memory of size $\delta Q_{l,i}^{2\mathrm{D}}$ (in bits), where $\delta$ is an integer, capable of storing the 2D and 3D FoVs of certain views. Let $p_{l,i}^{2\mathrm{D}} \in \{0,1\}$ denote the 2D FoV prefetching decision of HMD $l$ for FoV $i$, where $p_{l,i}^{2\mathrm{D}} = 1$ indicates that the 2D FoV $i$ is prefetched in memory of HMD $l$, and $p_{l,i}^{2\mathrm{D}} = 0$ otherwise. Let $p_{l,i}^{3\mathrm{D}} \in \{0,1\}$ denote the 3D FoV prefetching decision of HMD $l$ for FoV $i$, where $p_{l,i}^{3\mathrm{D}} = 1$ indicates that the 3D FoV $i$ is prefetched in memory of HMD $l$, and $p_{l,i}^{3\mathrm{D}} = 0$ otherwise. Under the memory constraint, the following condition must be satisfied:
$
\sum_{i=1}^V Q_{l,i}^{2\mathrm{D}} p_{l,i}^{2\mathrm{D}} 
    + \alpha Q_{l,i}^{2\mathrm{D}} p_{l,i}^{3\mathrm{D}}
    \leq \delta Q_{l,i}^{2\mathrm{D}}.
$
We assume that all 2D and 3D FoVs of all HMDs are prefetched to the MECS. This is reasonable because the memory of the MECS is much larger than that of any HMD.

Next, consider the projection rendering decision on HMD $l$. Suppose that HMD $l$ operates with a given CPU cycle frequency $f_l$ (unit: cycles/s), and the power consumed per rendering cycle is $\zeta f_l^2$, where $\zeta$ is the effective switched capacitance associated with the chip architecture and reflects the energy efficiency of the CPU of HMDs~\cite{sun2019communications}. We assume that the battery capacity of the HMDs is limited. To ensure that the user's VR viewing time does not fall below a required average duration, considering that other applications also consume resources and that battery lifetime should be extended, a long-term time-averaged power consumption constraint is introduced, denoted by $\bar{P}_l$ (unit: watts, W). Let $r_{l,i} \in \{0,1\}$ denote the rendering decision of HMD $l$ for FoV $i$, where $r_{l,i} = 1$ indicates that the projection from the 2D FoV to the 3D FoV $i$ is rendered on HMD $l$, and $r_{l,i} = 0$ otherwise. Under the power constraint, the following condition must hold \cite{liu2023rendered}:
$
\frac{\zeta f_l^2 Q_{l,i}^{2\mathrm{D}}o}{V \tau} 
\sum_{i=1}^V r_{l,i} \leq \bar{P}_l,
\label{eq2}
$
i, e.,
$
\sum_{i=1}^V r_{l,i} \leq \frac{V \bar{P}_l \tau}{\zeta f_l^2 Q_{l,i}^{2\mathrm{D}} o}, 
$
which actually represents the maximum number of FoVs that HMD $l$ can render under the request process, referred to as the rendering budget of HMD $l$. For simplicity, we assume that 
$
\frac{V \bar{P}_l \tau}{\zeta f_l^2 Q_{l,i}^{2\mathrm{D}} o}
$
is an integer.

Finally, let $\left(\boldsymbol{p}_l^{3\mathrm{D}}, \boldsymbol{p}_l^{2\mathrm{D}}, \mathbf{r}_l\right)$ denote the joint prefetching and rendering decisions of HMD $l$ for all FoVs, where $\boldsymbol{p}_l^{3\mathrm{D}} \triangleq (p_{l,i}^{3\mathrm{D}})_{i \in \mathcal{V}}$ represents the 3D FoV prefetching decision vector for all FoVs, and $\boldsymbol{p}_l^{2\mathrm{D}} \triangleq (p_{l,i}^{2\mathrm{D}})_{i \in \mathcal{V}}$ represents the 2D FoV prefetching decision vector that satisfies the memory constraint. The vector $\mathbf{r}_l \triangleq (r_{l,i})_{i \in \mathcal{V}}$ denotes the rendering decision vector.
Based on the joint prefetching and rendering decision 
$\left(\boldsymbol{p}_l^{3\mathrm{D}}, \boldsymbol{p}_l^{2\mathrm{D}}, \mathbf{r}_l\right)$,
the request of HMD $l$ for FoV $i \in \mathcal{V}$ can be served through one of the following four paths (illustrated by Fig. \ref{fig:paths}):
\begin{enumerate}
\renewcommand{\labelenumi}{\roman{enumi}:}  
\item {On-device 3D prefetching:}  
If $p_{l,i}^{3\mathrm{D}} = 1$, then the 3D FoV $i$ can be directly obtained from the memory of HMD $l$, requiring neither transmission nor rendering. Thus, the delay is negligible:
$
T_{l,i}^{\langle 1 \rangle}
 = 0 .
$
\item {On-device rendering with 2D prefetched:}  
If $p_{l,i}^{3\mathrm{D}} = 0$, $r_{l,i} = 1$, and $p_{l,i}^{2\mathrm{D}} = 1$, then HMD $l$ retrieves the 2D FoV $i$ from memory without transmission and renders it into a 3D FoV using its local CPU processor. The total delay is:
$
T_{l,i}^{\langle 2 \rangle}
 = \frac{Q_{l,i}^{2\mathrm{D}}o}{f_l}.
$
\item {On-device rendering with 2D downloaded:}  
If $p_{l,i}^{3\mathrm{D}} = 0$, $r_{l,i} = 1$, and $p_{l,i}^{2\mathrm{D}} = 0$, then HMD $l$ downloads the 2D FoV $i$ from the MECS and renders it locally. Thus, the total delay is:
$
T_{l,i}^{\langle 3 \rangle}
 = \frac{Q_{l,i}^{2\mathrm{D}}}{R_l}
           + \frac{Q_{l,i}^{2\mathrm{D}} o}{f_l},
$
where $\frac{Q_{l,i}^{2\mathrm{D}}}{R_l}$ is the transmission delay of downloading the 2D FoV via the wireless link, and $\frac{Q_{l,i}^{2\mathrm{D}} o}{f_l}$ is the rendering delay.
\item {Remote 3D transmission:}  
If $p_{l,i}^{3\mathrm{D}} = 0$ and $r_{l,i} = 0$, then HMD $l$ downloads the 3D FoV $i$ from the MECS. The total delay is:
$
T_{l,i}^{\langle 4 \rangle}
 = \frac{Q_{l,i}^{3\mathrm{D}}}{R_l}.
$
\end{enumerate}
\begin{figure}[h]
\centering
\includegraphics[width=0.5\textwidth]{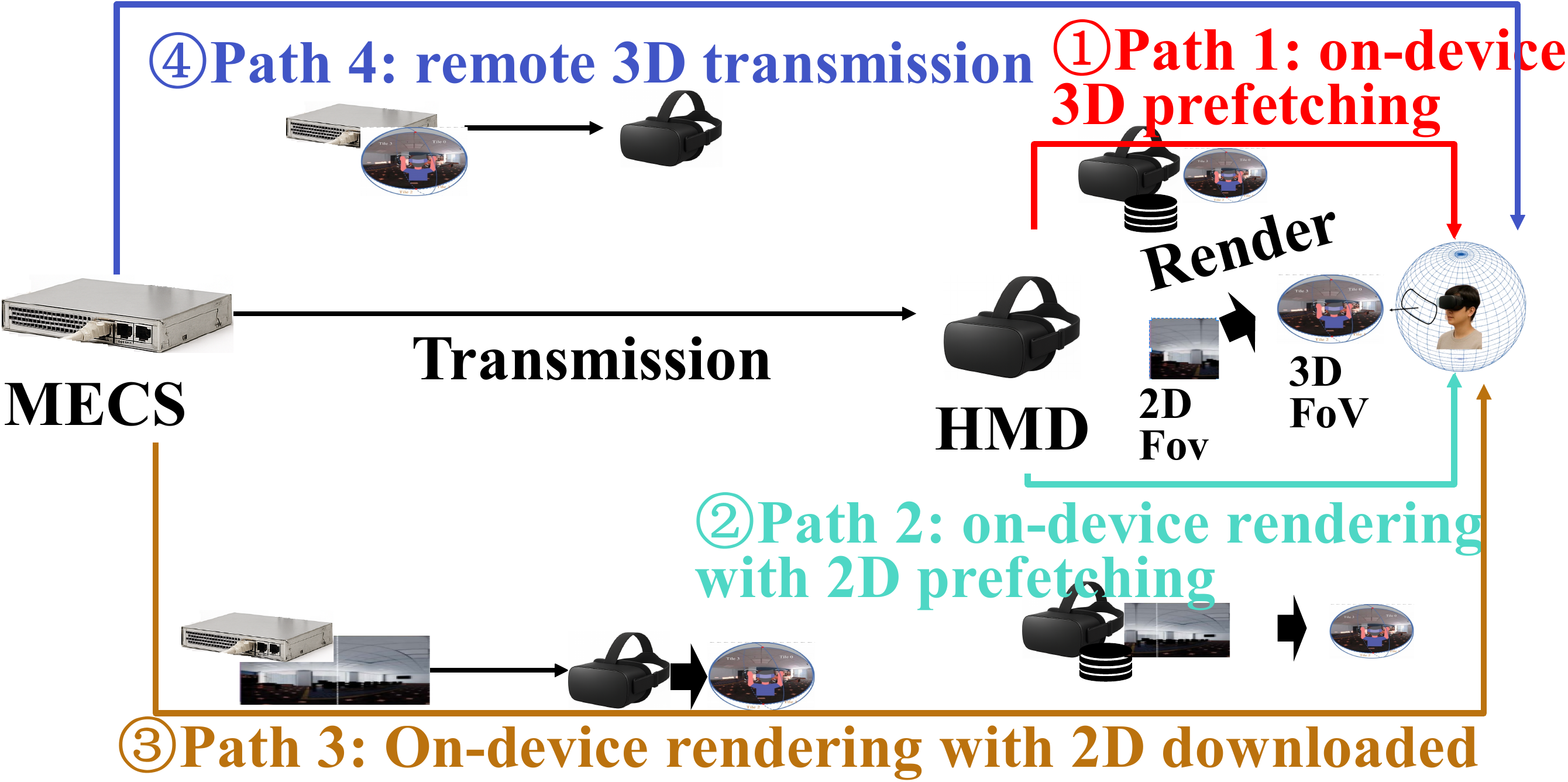}
\caption{Four projection paths.}
\label{fig:paths}
\end{figure}

\subsection{Problem Formulation}
For HMD $l$, the total latency to process FoV $i$ is given by
$T_{l,i} =  p_{l,i}^{3\mathrm{D}} T_{l,i}^{\langle 1 \rangle}
 + (1 - p_{l,i}^{3\mathrm{D}}) r_{l,i} p_{l,i}^{2\mathrm{D}} T_{l,i}^{\langle 2 \rangle}
 + (1 - p_{l,i}^{3\mathrm{D}}) r_{l,i} (1 - p_{l,i}^{2\mathrm{D}}) T_{l,i}^{\langle 3 \rangle}
  + (1 - p_{l,i}^{3\mathrm{D}})(1 - r_{l,i}) T_{l,i}^{\langle 4 \rangle}$
Substituting the four expressions of $T_{l,i}^{\langle j \rangle}, j\in \{1,2,3,4\}$, we have 
$T_{l,i}
=
(1 - p_{l,i}^{3\mathrm{D}})
\Bigg[
r_{l,i}
\left(
\frac{Q_{l,i}^{2\mathrm{D}} o}{f_l}
+
(1 - p_{l,i}^{2\mathrm{D}})
\frac{Q_{l,i}^{2\mathrm{D}}}{R_l(\mathbf a,\boldsymbol P)}
\right)
+(1 - r_{l,i})
\frac{Q_{l,i}^{3\mathrm{D}}}{R_l(\mathbf a,\boldsymbol P)}
\Bigg],
$
where $r_{l,i}\left(\frac{Q_{l,i}^{2\mathrm{D}} o}{f_l} + (1 - p_{l,i}^{2\mathrm{D}})\frac{Q_{l,i}^{2\mathrm{D}}}{R_l(\mathbf a,\boldsymbol P)}\right)$ 
represents the total delay incurred when on-device rendering is selected. If the 2D FoV is prefetched in memory, no transmission delay occurs; otherwise, an additional transmission delay of $\frac{Q_{l,i}^{2\mathrm{D}}}{R_l(\mathbf a,\boldsymbol P)}$ is included. $(1 - r_{l,i})\frac{Q_{l,i}^{3\mathrm{D}}}{R_l(\mathbf a,\boldsymbol P)}$ 
corresponds to the delay of downloading the 3D FoV from the MECS. Finally, the outer multiplicative factor 
$(1 - p_{l,i}^{3\mathrm{D}})$ 
indicates that these delays only apply when the 3D FoV is not prefetched in memory. 
 $R_l(\mathbf a,\boldsymbol P)$ denotes the downlink transmission rate of HMD $l$, which is determined by the RHS beamforming parameters: holographic pattern weights $\mathbf a$ and the feed power allocation $\boldsymbol P$: $\mathbf a \triangleq \{ \mathbf a_u \}_{u=1}^{U}$, $\mathbf a_u = \{ a_{l,u,k} \}_{l=1,k=1}^{L,K}$, and $\boldsymbol P \triangleq \{ \boldsymbol P_u \}_{u=1}^{U}$, $\boldsymbol P_u = \{ P_{u,k}\}_{k=1}^{K}$, $\mathrm{Tr}(\boldsymbol{P}_u\boldsymbol{P}_u^H)=P_{u}^\text{Tx}=\frac{P^\text{Tx}}{U}$.

We now formulate a joint optimization problem that minimizes the average FoV processing delay across all HMDs and FoVs by jointly optimizing the content prefetching $\left(\boldsymbol{p}_l^{3\mathrm{D}}, \boldsymbol{p}_l^{2\mathrm{D}}\right)$, rendering offloading $\mathbf{r}_l$, and RHS holographic beamforming $(\mathbf a,\boldsymbol P)$ decisions.

\begin{problem}[Joint delay minimization]
\label{prob:joint_delay}
\begin{equation}
\begin{aligned}
(\mathcal{P}_1):\;
&\min_{\substack{
\boldsymbol{p}_l^{3\mathrm{D}}, \boldsymbol{p}_l^{2\mathrm{D}}, \mathbf{r}_l,\\
\mathbf a,\,
\boldsymbol P
}}
\quad 
\frac{1}{LV}
\sum_{l=1}^{L}
\sum_{i=1}^{V}
\pi_{l,i}T_{l,i}
\\[4pt]
\mathrm{s.t.}\quad
& \text{(C1)}~
\sum_{i=1}^{V}
\left(
Q_{l,i}^{2\mathrm{D}} p_{l,i}^{2\mathrm{D}}
+
Q_{l,i}^{3\mathrm{D}} p_{l,i}^{3\mathrm{D}}
\right)
\le
\delta,
\quad
\forall l \in \mathcal L,
\\[4pt]
& \text{(C2)}~
\sum_{i=1}^{V}
\pi_{l,i}
\frac{\zeta f_l^2 Q_{l,i}^{2\mathrm{D}} o}{\tau}
r_{l,i}
\le
\bar P_l,
\quad
\forall l \in \mathcal L,
\\[4pt]
& \text{(C3)}~
\frac{Q_{l,i}^{2\mathrm{D}}}{R_l(\mathbf a,\boldsymbol P)}
+
\frac{Q_{l,i}^{2\mathrm{D}} o}{f_l}
\le
\tau,
\quad
\text{if } r_{l,i}=1,\; p_{l,i}^{2\mathrm{D}}=0,
\\[4pt]
& \text{(C4)}~
\frac{Q_{l,i}^{3\mathrm{D}}}{R_l(\mathbf a,\boldsymbol P)}
\le
\tau,
\quad
\text{if } r_{l,i}=0,
\\[4pt]
& \text{(C5)}~
p_{l,i}^{2\mathrm{D}},
\;
p_{l,i}^{3\mathrm{D}},
\;
r_{l,i}
\in
\{0,1\},
\quad
\forall l \in \mathcal L,\; i \in \mathcal V,
\\[4pt]
& \text{(C6)}~
a_{l,u,k} \ge 0,
\quad
\sum_{l=1}^{L}
\sum_{k=1}^{K}
a_{l,u,k}
= 1,
\forall u \in \{1,\ldots,U\},
\\[4pt]
& \text{(C7)}~
P_{u,k} \ge 0, \; \boldsymbol P_u = \{ P_{u,k}\}_{k=1}^{K}, 
\end{aligned}
\label{eq:P1_final}
\end{equation}
\end{problem}

\section{Equivalent Reformulation under the Homogeneity and Staticity Assumption}
We consider the uniform distribution for all FoVs, i.e., $\pi_{l,i} = \frac{1}{V}$ for each $i \in \mathcal{V}$. We also assume that for any HMD $l$,   all FoVs have equal size, i.e.,
$
Q_{l,i}^{2\mathrm{D}} = Q^{2\mathrm{D}}, 
$
$ 
Q_{l,i}^{3\mathrm{D}} = Q^{3\mathrm{D}} = \alpha Q^{2\mathrm{D}},
$
and  mark $f_l$ as $f$, $p_{l,i}^{2\mathrm{D}}$ as $p_{i}^{2\mathrm{D}}$, $p_{l,i}^{3\mathrm{D}}$ as $p_{i}^{3\mathrm{D}}$, $\bar P_l$ as $\bar P$, and $r_{l,i}$ as $r_i$. Under the staticity assumption, the HMDs remain stationary, such that the geometric relationship between the RHS and the HMDs does not change. Consequently, the RHS beamforming configuration is fixed and the downlink transmission rate $R_l$ remains unchanged (marked as $R$). Then, the total delay objective of HMD $l$ simplifies to 
$
\sum_{i=1}^{V}
(1 - p_i^{3\mathrm{D}})
\left[
r_i\left(
\frac{Q^{2\mathrm{D}} o}{f}
+
(1 - p_i^{2\mathrm{D}})\frac{Q^{2\mathrm{D}}}{R}
\right)
+
(1 - r_i)\frac{\alpha Q^{2\mathrm{D}}}{R}
\right]
=
\sum_{i=1}^{V}
(1 - p_i^{3\mathrm{D}})
Q^{2\mathrm{D}}
\left[
r_i\left(
\frac{o}{f}
+
\frac{1 - p_i^{2\mathrm{D}}}{R}
\right)
+
(1 - r_i)\frac{\alpha}{R}
\right].
$

The corresponding constraints become:

\noindent $\mathrm{s. t.}$
\begin{equation}
\begin{aligned}
\text{(C1e)}\quad &
\sum_{i=1}^{V} p_i^{2\mathrm{D}}
+ \alpha p_i^{3\mathrm{D}}
\le \delta,
\\
\text{(C2e)}\quad &
\sum_{i=1}^{V} r_i
\le
\frac{V \bar{P} \tau}{\zeta f^2 Q^{2\mathrm{D}} o}
\triangleq Q_{\mathrm{max}},
\\
\text{(C3e)}\quad &
\frac{Q^{2\mathrm{D}}}{R}
+
\frac{Q^{2\mathrm{D}} o}{f}
< \tau,
\\
\text{(C4e)}\quad &
\frac{Q^{3\mathrm{D}}}{R}
\le \tau.
\\
\text{(C5e)}\quad &
p_{i}^{2\mathrm{D}},
p_{i}^{3\mathrm{D}},
r_{i}
\in \{0,1\},
\quad
\forall l \in \mathcal{L},
 i \in \mathcal{V} .
\end{aligned}
\end{equation}

\begin{proposition}[3D prefetching eliminates the need for 2D prefetching and rendering]
For any FoV $i \in \mathcal{V}$, if $p_i^{3\mathrm{D}} = 1$, then in the optimal solution we have 
$p_i^{2\mathrm{D}} = 0$ and $r_i = 0$.
\end{proposition}

\begin{proof}
If $p_i^{3\mathrm{D}} = 1$, the 3D FoV $i$ is locally prefetched, and thus no transmission or rendering is required; the equivalent delay is approximately zero.  
Setting $p_i^{2\mathrm{D}} = 1$ and $r_i = 1$ in this case only consumes additional memory or computation resources and increases delay.  
Therefore, the optimal solution must satisfy $p_i^{2\mathrm{D}} = 0$ and $r_i = 0$. 
\end{proof}

\begin{proposition}[2D prefetching must be accompanied by local rendering]
For any FoV $i \in \mathcal{V}$, if $r_i = 0$, then it is without loss of optimality to set $p_i^{2\mathrm{D}} = 0$.
\end{proposition}

\begin{proof}
If $r_i = 0$, no local rendering is performed for FoV $i$, and thus the 2D FoV is never used.  
The FoV must be served by remote 3D transmission, whose delay is independent of the 2D prefetching decision.  
Assigning $p_i^{2\mathrm{D}} = 1$ would only consume unnecessary memory, increasing delay.  
Thus, setting $p_i^{2\mathrm{D}} = 0$ is optimal.  
\end{proof}

\begin{proposition}[the memory constraint is tight at the optimum]
At the optimal solution, the memory constraint satisfies
$
\sum_{i=1}^V p_i^{2\mathrm{D}} + \alpha p_i^{3\mathrm{D}} = \delta,
$
without loss of optimality.
\end{proposition}

\begin{proof}
Consider any feasible solution that does not fully utilize the available memory. Suppose there exists a FoV $j$ currently served through the remote 3D transmission path (i.e., $p_j^{3\mathrm{D}} = 0$ and $r_j = 0$), whose delay is $\frac{Q^{3\mathrm{D}}}{R}$. If there is still remaining memory, it can be used to store the 3D FoV (i.e., setting $p_j^{3\mathrm{D}} = 1$), thereby reducing its delay from $\frac{Q^{3\mathrm{D}}}{R}$ to approximately zero. This strictly decreases the total objective. Therefore, at optimality, it is impossible to have unused memory while the constraint is not tight. Hence, the memory constraint must hold with equality.  
\end{proof}

Combining Propositions 1, 2, and 3, we can obtain an equivalent problem in terms of aggregated variables. Define the aggregated decision variables for an HMD as $p^{3\mathrm{D}} \triangleq \sum_{i=1}^V p_i^{3\mathrm{D}}$, $p^{2\mathrm{D}} \triangleq \sum_{i=1}^V p_i^{2\mathrm{D}}$, $r \triangleq \sum_{i=1}^V r_i$, we can assume ``without loss of optimality'' the following structure:
\begin{itemize}
    \item FoVs $i = 1, \ldots, p^{3\mathrm{D}}$: on-device 3D prefetching;
    \item The next $q \triangleq \min\{p^{2\mathrm{D}}, r\}$ FoVs:  on-device rendering with 2D prefetched;
    \item The next $r - q$ FoVs: on-device rendering with 2D downloaded;
    \item The remaining $V - p^{3\mathrm{D}} - r$ FoVs: remote 3D transmission.
\end{itemize}

Let $T^{\langle 2 \rangle}
 \triangleq \frac{Q^{2\mathrm{D}_o}}{f}$, $T^{\langle 3 \rangle}
  \triangleq \frac{Q^{2\mathrm{D}}}{R} + \frac{Q^{2\mathrm{D}o}}{f}$, $T^{\langle 4 \rangle}
  \triangleq \frac{Q^{3\mathrm{D}}}{R} = \frac{\alpha Q^{2\mathrm{D}}}{R}$. The average delay can then be written as:
$
\bar{T} = \frac{1}{V} \left[ q T^{\langle 2 \rangle}
  + (r - q) T^{\langle 3 \rangle}
  + \left(V - p^{3\mathrm{D}} - r\right) T^{\langle 4 \rangle}
  \right],
$
where $q = \min\{p^{2\mathrm{D}}, r\}$. By Proposition 3, the memory constraint holds as equality:
$
p^{2\mathrm{D}} + \alpha p^{3\mathrm{D}} = \delta
 \Rightarrow 
p^{3\mathrm{D}} = \frac{\delta -p^{2\mathrm{D}}}{\alpha}.
$
The average power constraint becomes:
$
r \leq Q_{\mathrm{max}}, 
Q_{\mathrm{max}} \triangleq \frac{V \bar{P} \tau}{\zeta f^2 Q^{2\mathrm{D}o}}.
$
We also assume that
$
\frac{\delta}{\alpha} + Q_{\mathrm{max}} \leq V,
$
so that the total number of FoVs served by prefetching and local rendering does not exceed $V$.

We now obtain the equivalent problem:  
\newtheorem*{problem1'}{\textbf{Problem 1'} \textnormal{(equivalent $(\mathcal{P}_1)$ under homogeneity and staticity assumption)}}
\begin{problem1'}
\begin{equation}
\begin{aligned}
(\mathcal{P}_1') &\min_{p^{2\mathrm{D}},p^{3\mathrm{D}}, r}  \bar{T}\left(p^{2\mathrm{D}}, p^{3\mathrm{D}}, r\right)  = \min_{p^{2\mathrm{D}}, r} \frac{1}{V} \\
 &\left[ q T^{\langle 2 \rangle}
 + (r - q) T^{\langle 3 \rangle}
 + \left(V - \frac{\delta - p^{2\mathrm{D}}}{\alpha} - r\right) T^{\langle 4 \rangle}
 \right] \\
\mathrm{s.t.}\quad &  p^{2\mathrm{D}} \in \{0, 1, \dots, \delta\}, p^{3\mathrm{D}} = \frac{\delta - p^{2\mathrm{D}}}{\alpha}, p^{3\mathrm{D}} + r \leq V, \\
& r \in \{0, 1, \dots, Q_{\mathrm{max}}\}, q = \min \left\{p^{2\mathrm{D}}, r \right\} 
\end{aligned}
\end{equation}
\end{problem1'}

The core comparison is between the delay of ``on-device rendering with 2D downloaded'' (path 3) and ``remote 3D transmission'' (path 4):
$
T^{\langle 3 \rangle}
  \stackrel{?}{\gtrless}  T^{\langle 4 \rangle}
  \Leftrightarrow \frac{Q^{2\mathrm{D}}}{R} + \frac{Q^{2\mathrm{D}} o}{f} \stackrel{?}{\gtrless} \frac{\alpha Q^{2\mathrm{D}}}{R}.
$
Dividing both sides by $Q^{2\mathrm{D}}$:
$
\frac{1}{R} + \frac{o}{f} \stackrel{?}{\gtrless} \frac{\alpha}{R} \Longleftrightarrow \frac{o}{f} \stackrel{?}{\gtrless} \frac{\alpha - 1}{R}.
$

Thus, we define two zones.
\begin{definition}[remote 3D transmission priority zone] The zone where remote 3D transmission is prioritized is defined by the condition:
$
T^{\langle 3 \rangle}
 \geq T^{\langle 4 \rangle}
 \Leftrightarrow \frac{o}{f} \geq \frac{\alpha - 1}{R},
$
where local rendering (without prefetching) is less efficient or slower than remote 3D transmission.
\end{definition}

\begin{definition}[on-device rendering priority zone]
The zone where local rendering is prioritized is defined by the condition:
$
T^{\langle 3 \rangle}
  < T^{\langle 4 \rangle}
  \Longleftrightarrow o \leq \frac{\alpha - 1}{f},
$
where local rendering is faster than remote 3D transmission.
\end{definition}

The following Proposition \ref{theorem1} and Proposition \ref{theorem2} provide the optimal strategies for $\mathcal{P}_1'$.

\begin{proposition}[the optimal strategy for the remote 3D transmission priority zone is closed-form]
\label{theorem1}
Assume that for a certain HMD, the constants $R$ and $f$ satisfy
$
\frac{o}{f} \geq \frac{\alpha - 1}{R}, \quad \text{i.e.,}  T^{\langle 3 \rangle}
  \geq T^{\langle 4 \rangle}
 ,
$
and the memory and power budget satisfy
$
\frac{\delta}{\alpha} + Q_{\mathrm{max}} \leq V.
$
Then, the optimal joint strategy for $\mathcal{P}_1'$ is $p^{2\mathrm{D}^{*}} = \min \left\{\delta, Q_{\mathrm{max}}\right\}$, 
$r^* = p^{2\mathrm{D}^{*}}$, $p^{3\mathrm{D}^{*}} = \frac{\delta - p^{2\mathrm{D}^{*}}}{\alpha}$.
The corresponding minimum average delay is:
$
\bar{T}^* = T^{\langle 4 \rangle}
  - \frac{T^{\langle 4 \rangle}
 }{V} \cdot \frac{\delta}{\alpha} - \frac{T^{\langle 4 \rangle}
  - T^{\langle 2 \rangle}
 }{V} \cdot p^{2\mathrm{D}^{*}},
$
where $T^{\langle 2 \rangle}
  = \frac{Q^{2\mathrm{D}o}}{f}$ and $T^{\langle 4 \rangle}
  = \frac{Q^{3\mathrm{D}}}{R}$.
\end{proposition}

\begin{proof}
See Appendix A.
\end{proof}

\begin{proposition}[the optimal strategy for the on-device rendering priority zone is closed-form]
\label{theorem2}
Assume that for a certain HMD, the constants $R$ and $f$ satisfy
$
\frac{o}{f} < \frac{\alpha - 1}{R}, \quad \text{i.e., }  T^{\langle 3 \rangle}
  < T^{\langle 4 \rangle}
 ,
$
and that the memory and power budget satisfy
$
\frac{\delta}{\alpha} + Q_{\mathrm{max}} \leq V.
$
Then, the optimal joint strategy for $\mathcal{P}_1'$ is $p^{2\mathrm{D}^*} = \min \left\{\delta, Q_{\mathrm{max}}\right\}$, $r^* = Q_{\mathrm{max}}$, $p^{3\mathrm{D}^*} = \frac{\delta - p^{2\mathrm{D}^*}}{\alpha}$. The corresponding minimum average delay is:
$
\bar{T}^* = T^{\langle 4 \rangle}
  - \frac{T^{\langle 4 \rangle}
 }{V} \cdot \frac{\delta}{\alpha} - \frac{T^{\langle 4 \rangle}
  - T^{\langle 2 \rangle}
 }{V} p^{2\mathrm{D}^*} - \frac{T^{\langle 4 \rangle}
  - T^{\langle 3 \rangle}
 }{V} \left(Q_{\mathrm{max}} - p^{2\mathrm{D}^*}\right).
$
\end{proposition}
\begin{proof}
See Appendix B.
\end{proof}  

For the homogeneous and static multi-user case, since each FoV is homogeneous, the memory and computation constraints are independent of each other for HMDs, and the downlink transmission rate remains constant, we solve by:
\begin{enumerate}
\renewcommand{\labelenumi}{\roman{enumi}:}  
    \item For each HMD $ l$, we can write a separate Problem 1'.
    \item Based on the parameters of the HMD $\left(R_l, f_l, \bar{P}_l\right)$, we determine whether it belongs to the remote 3D transmission priority zone or the on-device rendering priority zone.
    \item Apply Proposition \ref{theorem1} or Proposition \ref{theorem2} accordingly to obtain $\left(p_l^{3 \mathrm{D}^*}, p_l^{2 \mathrm{D}^*}, r_l^*\right)$.
    \item The optimal joint strategy for the entire multi-user system is the collection of optimal strategies for all HMDs.
\end{enumerate}

\section{Equivalent Reformulation under the Heterogeneity and Non-Staticity  Assumption}
In heterogeneous setting, FoV parameters differ across HMDs, while the THz downlink rate depends on RHS beamforming. Because an RHS can reconfigure its holographic beams in real time, user mobility (e.g., an HMD moving indoors) requires frequent short-timescale beam updates, whereas prefetching and rendering decisions evolve on a much slower timescale. This natural separation motivates decomposing the joint problem into a long-timescale prefetching–rendering subproblem ($\mathcal{P}$1-1) using long-term equivalent rates, and a short-timescale RHS beamforming subproblem ($\mathcal{P}$1-2) that rapidly adapts to user movement. $\mathcal{P}$1-1 and $\mathcal{P}$1-2 are coupled via the downlink rate vector ${\mathbf{R}} = ({R}_1, \ldots, {R}_L)$.

\newtheorem*{problem1-1}{\textbf{Problem 1-1} \textnormal{(prefetching-rendering optimization (given $\bar{\mathbf{R}}$))}}
\begin{problem1-1}
We assume that the physical layer has received a long-term equivalent rate vector ${\mathbf{R}}$ on a faster timescale, then:
\begin{equation}
(\mathcal{P}1\text{-}1):  \min_{\boldsymbol{p}, \mathbf{r}} \bar{T} = \min_{\boldsymbol{p}, \mathbf{r}} \frac{1}{L} \sum_{l=1}^L \sum_{i=1}^V \pi_{l, i} T_{l, i}.
\end{equation}
\end{problem1-1}

We introduce the path selection variable $x_{l, i, j}$ for each HMD $l$, FoV $i$, and path $j \in \{1, 2, 3, 4\}$, where $x_{l, i, j} \in \{0, 1\}$, and if $x_{l, i, j} = 1$, it means FoV $(l, i)$ follows path $j$, otherwise it does not. For each $(l, i)$, exactly one path must be selected:
$
\sum_{j=1}^4 x_{l, i, j} = 1,  \forall l,  i.
$
The delays corresponding to the four paths are: $T_{l, i}^{\langle 1 \rangle}
 = 0$,  $T_{l, i}^{\langle 2 \rangle}
 = \frac{Q_{l, i}^{2 \mathrm{D}} o}{f_l}$, $T_{l, i}^{\langle 3 \rangle}
 = \frac{Q_{l, i}^{2 \mathrm{D}}}{R_l} + \frac{Q_{l, i}^{2 \mathrm{D}} o}{f_l}$, $ T_{l, i}^{\langle 4 \rangle}
 = \frac{Q_{l, i}^{3 \mathrm{D}}}{R_l}.$ Thus, the average delay for HMD $l$ can be written as:
$
\sum_{{i}=1}^{V} \pi_{l, i} \sum_{{j}=1}^4 T_{l, i}^{\langle j \rangle}
 {x}_{l, {i}, {j}}
$.
Adding up for all users and dividing by $L$, we get the total delay target:
$
\bar{{T}} = \frac{1}{{L}} \sum_{{l}=1}^{{L}} \sum_{{i}=1}^{{V}} \pi_{l, i} \sum_{{j}=1}^4 T_{l, i}^{\langle j \rangle}
 {x}_{l, {i}, {j}}
$.

For each path, define the memory consumption $ C_{l, i, j}^{(\mathrm{m})} $ and  power consumption $C_{l, i, j}^{(\mathrm{p})}$:
\begin{equation}
C_{l, i, j}^{(\mathrm{m})} \triangleq
\begin{cases}
Q_{l, i}^{3 \mathrm{D}}, & j=1, \\
Q_{l, i}^{2 \mathrm{D}}, & j=2, \\
0, & j=3,4,
\end{cases}
C_{l, i, j}^{(\mathrm{p})} \triangleq
\begin{cases}
\pi_{l, i} \frac{\zeta f_l^2 D_{l, i}^{2 \mathrm{D}} o}{\tau}, & j=2,3, \\
0, & j=1,4.
\end{cases}
\end{equation}

The problem is equivalent to the following form:
\newtheorem*{problem1-1-H1}{\textbf{Problem 1-1-H1} \textnormal{(MMKP problem)}} 
\begin{problem1-1-H1} 
\begin{equation}
\begin{array}{ll}
(\mathcal{P}\text{1-1-H1})&\\
\min_{\left\{{x}_{l, i,j}\right\}} & \bar{T} = \frac{1}{{L}} \sum_{{l}=1}^{{L}} \sum_{{i}=1}^{{V}} \pi_{l, i} \sum_{{j}=1}^4 T_{l, i}^{\langle j \rangle}
 {x}_{l, {i}, {j}} \\
\mathrm{s. t.} & \sum_{{i}=1}^{{V}} \sum_{{j}=1}^4 {C}_{l, {i}, {j}}^{(\mathrm{m})} {x}_{1, {i}, {j}} \leq \delta Q_{l,i}^{2 \mathrm{D}}, \ \forall {l \in \mathcal{L}}, \\
& \sum_{{i}=1}^{{V}} \sum_{{j}=1}^4 {C}_{l, {i}, {j}}^{(\mathrm{p})} {x}_{1, {i}, {j}} \leq \bar{P}_l,  \forall {l \in \mathcal{L}}, \\
& \sum_{{j}=1}^4 {x}_{l, {i}, {j}} = 1, \ \forall {l \in \mathcal{L}}, {i\in \mathcal{V}}, \\
& {x}_{l, {i}, {j}} \in \{0,1\},  \forall {l\in \mathcal{L}}, {i\in \mathcal{V}}, {j\in \{1,2,3,4\}}.
\end{array}
\end{equation}
\end{problem1-1-H1}
This is a standard Multiple-Choice Multi-Dimensional Knapsack Problem (MMKP): each $(l, i)$ needs to choose one option out of four (multiple-choice); there are two resource dimensions, memory and power (multi-dimensional).

Let path 4 (remote 3D transmission) be the baseline, and define the baseline delay for field \( i \) as \( {T}_{l, {i}}^{\langle 4 \rangle}
\). The ``delay gain'' (reduction) for path $j$ is defined as:
\begin{equation}
g_{l, i, j} \triangleq 
\begin{cases} 
\pi_{l, i} \left( T_{l, i}^{\langle 4 \rangle}
 - T_{l, i}^{\langle 1 \rangle}
 \right), & j=1, \\
\pi_{l, i} \left( T_{l, i}^{\langle 4 \rangle}
 - T_{l, i}^{\langle 2 \rangle}
 \right), & j=2, \\
\pi_{l, i} \left( T_{l, i}^{\langle 4 \rangle}
 - T_{l, i}^{\langle 3 \rangle}
 \right), & j=3, \\
0, & j=4.
\end{cases}
\end{equation}
Thus, the total average delay can be written as:
$
\bar{T} = \frac{1}{L} \sum_{l, i} \pi_{l, i} T_{l, i}^{\langle 4 \rangle}
 - \frac{1}{L} \sum_{l, i} \sum_{j=1}^4 g_{l, i, j} x_{l, i, j}
$.
The first term is a constant that does not depend on the decision \( x_{l, i, j} \); therefore, minimizing $\bar{T}$ is equivalent to maximizing the total delay gain:
$
\max_{\left\{ x_{l, i, j} \right\}} \sum_{l, i} \sum_{j=1}^4 g_{l, i, j} x_{l, i, j}
$.

Next, we introduce $\mathcal{P}$1-1-H2.
\newtheorem*{problem1-1-H2}{\textbf{Problem 1-1-H2} \textnormal{(binary relaxation \& IQP problem)}} 
\begin{problem1-1-H2}
\begin{equation}
\begin{aligned}
(\mathcal{P}\text{1-1-H2}) & \min _{\left\{x_{l, i, j}\right\}} \sum_{l, i} \sum_{j=1}^4  -g_{l, i, j} x_{l, i, j} \\
 \mathrm{s. t.} \quad &\sum_j x_{l, i, j} = 1, \forall {l\in \mathcal{L}}, {i\in \mathcal{V}}, {j\in \{1,2,3,4\}}\\
& 0 \leq x_{l, i, j} \leq 1, \forall {l\in \mathcal{L}}, {i\in \mathcal{V}}, {j\in \{1,2,3,4\}} \\
& \sum_{l, i, j} x_{l, i, j} \left( 1 - x_{l, i, j} \right) \leq 0
\end{aligned}
\end{equation}
\end{problem1-1-H2}

Here, we  relax the binary decision variables to the continuous domain, i.e., $x_{l,i,j} \in [0,1],  \forall l,i,j$. Under this relaxation, the original binary constraint can be equivalently enforced through a quadratic inequality. Specifically, the integrality condition $x_{l,i,j} \in \{0,1\}$ is equivalent to
$
x_{l,i,j}(1-x_{l,i,j}) = 0,
$
which can be further rewritten in the aggregated form
$
\sum_{l,i,j} x_{l,i,j}\left(1-x_{l,i,j}\right) \le 0.
$
This reformulation converts the original binary constraint into the combination of a linear constraint and a concave quadratic constraint (since $x_{l,i,j}(1-x_{l,i,j})$ is concave), resulting in an integer quadratic programming (IQP) structure that is amenable to DC  techniques.

\begin{algorithm}[t]
\caption{Heuristic for $\mathcal{P}$1-1-H1 (MMKP)}
\label{alg:MMKP_heuristic}
\begin{algorithmic}[1]
\REQUIRE $\{\pi_{l,i}\}$, $\{T_{l,i}^{\langle j \rangle}
\}$, resource costs $\{C_{l, i, j}^{(\mathrm{m})},C_{l, i, j}^{(\mathrm{p})}\}$, memory budget $\delta$, power budgets $\{\bar P_l\}$, avoid division-by-zero constant $\epsilon$, memory-cost weighting coefficient $\alpha_1$, power-cost weighting coefficient $\alpha_2$
\ENSURE Feasible binary solution $\{x_{l,i,j}\}$
\STATE  For all $(l,i)$ set $x_{l,i,4}\gets1$ (remote 3D transmission baseline); others 0.  Record initial resource usage 
${\rm mem}_l$ and ${\rm pow}_l$.

\FOR{each $(l,i)$ and $j\in\{1,2,3\}$}
    \STATE Delay gain: $g_{l,i,j} \gets \pi_{l,i}(T_{l,i}^{\langle 4 \rangle}
 - T_{l,i}^{\langle j \rangle}
)$
    \STATE Score: $\rho_{l,i,j} \gets g_{l,i,j}/(\epsilon + \alpha_1 C_{l, i, j}^{(\mathrm{m})} + \alpha_2 C_{l, i, j}^{(\mathrm{p})})$
\ENDFOR
\STATE Sort all candidates $(l,i,j)$ in descending order of $\rho_{l,i,j}$.

\FOR{each sorted $(l,i,j)$}
    \STATE Let $j'$ be the current path of $(l,i)$; compute tentative resource usage
    $
    \tilde {\rm mem}_l = {\rm mem}_l - C_{l, i, j'}^{(\mathrm{m})}+ C_{l, i, j}^{(\mathrm{m})},
    \tilde {\rm pow}_l = {\rm pow}_l - C_{l, i, j'}^{(\mathrm{p})}+ C_{l, i, j}^{(\mathrm{p})}.
    $
    \IF{$\tilde {\rm mem}_l \le \delta$ and $\tilde {\rm pow}_l \le \bar P_l$}
        \STATE $x_{l,i,j'}\gets0, x_{l,i,j}\gets1$; update ${\rm mem}_l,{\rm pow}_l$.
    \ENDIF
\ENDFOR

\RETURN $\{x_{l,i,j}\}$
\end{algorithmic}
\end{algorithm}

\begin{algorithm}[t]
\caption{Binary relax and IQP for $\mathcal{P}$1-1-H2}
\label{alg:IQP_relaxation}
\begin{algorithmic}[1]
\REQUIRE Gains $\{g_{l,i,j}\}$, costs $\{C_{l, i, j}^{(\mathrm{m})},C_{l, i, j}^{(\mathrm{p})}\}$, budgets $\delta Q_{l,i}^{2\mathrm{D}}$, $\{\bar P_l\}$, penalty $\lambda>0$
\ENSURE Relaxed solution $\{x_{l,i,j}\}\in[0,1]$
\STATE Use Algorithm~\ref{alg:MMKP_heuristic} to obtain a feasible binary point $\{x_{l,i,j}^{(0)}\}$.
\STATE Relax $x_{l,i,j}\in\{0,1\}$ to $0\le x_{l,i,j}\le1$, $\sum_{j}x_{l,i,j}=1$, and keep all memory/power constraints.
\STATE Consider the penalized IQP objective
$
f(x)=\sum_{l,i,j}\Big(-g_{l,i,j}x_{l,i,j}-\lambda x_{l,i,j}(x_{l,i,j}-1)\Big),
$
whose concave part is $-\lambda x_{l,i,j}(x_{l,i,j}-1)$. Linearize it at $x^{(0)}$:
$
-\lambda x_{l,i,j}(x_{l,i,j}-1)= 
-\lambda\Big(x_{l,i,j}^{(0)}(x_{l,i,j}^{(0)}-1)
+(2x_{l,i,j}^{(0)}-1)(x_{l,i,j}-x_{l,i,j}^{(0)})\Big).
$
\STATE Drop constants independent of $x$ and solve the following convex QP:
$
\min_{\{x_{l,i,j}\}} 
\sum_{l,i,j}\Big(-g_{l,i,j}x_{l,i,j}
-\lambda(2x_{l,i,j}^{(0)}-1)x_{l,i,j}\Big)
$
subject to
$
\sum_{j}x_{l,i,j}=1,
0\le x_{l,i,j}\le1,
\sum_{i,j}C_{l, i, j}^{(\mathrm{m})}x_{l,i,j}\le\delta Q_{l,i}^{2\mathrm{D}},\;
\sum_{i,j}C_{l, i, j}^{(\mathrm{p})}x_{l,i,j}\le\bar P_l,\ \forall l.
$
\STATE Denote the optimizer by $\{x_{l,i,j}\}$ and use it as the relaxed H2 solution (and, if needed, as initialization for Algorithm~\ref{alg:CCP}).
\RETURN $\{x_{l,i,j}\}$
\end{algorithmic}
\end{algorithm}

\begin{algorithm}[t]
\caption{DCA/CCCP algorithm for $\mathcal{P}$1-1-H3}
\label{alg:CCP}
\begin{algorithmic}[1]
\REQUIRE Gains $\{g_{l,i,j}\}$, penalty $\lambda>0$, tolerance $\varepsilon>0$, max iterations $T_{\max}$
\ENSURE Approximate binary solution $\{x_{l,i,j}\}$
\STATE Obtain an initial feasible point $\{x_{l,i,j}^{(0)}\}$ (e.g., from Algorithm~\ref{alg:MMKP_heuristic}); set $t \gets 0$.
\WHILE{$t < T_{\text{max}}$}
    \STATE For all $(l,i,j)$, compute the gradient of the concave term at $x^{(t)}$:
    $
    \nabla_{x_{l,i,j}} \big[-\lambda x_{l,i,j}(x_{l,i,j}-1)\big]\Big|_{x^{(t)}} 
    = -\lambda\big(2x_{l,i,j}^{(t)}-1\big).
    $
     Build the convex surrogate objective by linearizing the concave part:
    $
    \tilde f^{(t)}(x) 
    = \sum_{l,i,j}\Big(-g_{l,i,j}x_{l,i,j} 
    -\lambda(2x_{l,i,j}^{(t)}-1)\,x_{l,i,j}\Big),
    $
    where constant terms w.r.t.\ $x$ are dropped.
    \STATE Solve the following convex LP/QP subproblem:
    $
    \min_{\{x_{l,i,j}\}}~\tilde f^{(t)}(x)
    $
    $
    \text{s.t. } \sum_{j} x_{l,i,j} = 1, 0 \le x_{l,i,j} \le 1, 
    \text{and memory/power constraints for all } l,i,j,
    $
    and denote the solution by $\{x_{l,i,j}^{(t+1)}\}$.
    \IF{$\max_{l,i,j}\big|x_{l,i,j}^{(t+1)} - x_{l,i,j}^{(t)}\big| \le \varepsilon$}
        \STATE \textbf{break}
    \ENDIF
    \STATE $t \gets t + 1$
\ENDWHILE
\STATE {For each $(l,i)$}, (rounding)$j^* \gets \arg\max_{j} x_{l,i,j}^{(t)}$; set $x_{l,i,j^\star} \gets 1$, $x_{l,i,j} \gets 0$ for all $j\neq j^*$
\RETURN $\{x_{l,i,j}\}$
\end{algorithmic}
\end{algorithm}

Now, we can add a penalty term $\lambda x_{l,i,j}(1-x_{l,i,j})$.

\newtheorem*{problem1-1-H3}{\textbf{Problem 1-1-H3} \textnormal{(penalty term \& CCP problem)}}
\begin{problem1-1-H3}
\begin{equation}
\begin{aligned}
(\mathcal{P}\text{1-1-H3})& \min _{\left\{x_{l, i, j}\right\}} \sum_{l, i} \sum_{j=1}^4 \left( -g_{l, i, j} x_{l, i, j} - \lambda  x_{l, i, j} \left( x_{l, i, j} - 1 \right) \right) \\
 \mathrm{s. t.} \quad &\sum_j ^4 x_{l, i, j} = 1, \forall {l\in \mathcal{L}}, {i\in \mathcal{V}} \\
& 0 \leq x_{l, i, j} \leq 1, \forall {l\in \mathcal{L}}, {i\in \mathcal{V}}, {j\in \{1,2,3,4\}}
\end{aligned}
\end{equation}
where $\lambda > 0 $ is the penalty parameter. 
\end{problem1-1-H3}

The objective function consists of a linear term $-g_{l, i, j} x_{l, i, j}$ and a concave quadratic penalty term $-\lambda x_{l, i, j}\left(x_{l, i, j}-1\right)$; hence, it has a difference-of-convex (DC) structure and can be solved using difference-of-convex algorithm (DCA) or convex-concave procedure (CCCP).

A feasible binary solution is first obtained via a greedy delay-gain–to–resource-cost heuristic (Algorithm~\ref{alg:MMKP_heuristic}), followed by binary relaxation with one-shot DCA/CCCP refinement (Algorithm~\ref{alg:IQP_relaxation}) and a full DCA/CCCP solver with rounding (Algorithm~\ref{alg:CCP}).


\begin{figure*}
\begin{equation}
(\mathcal{P}1\text{-}2):    \max_{\mathbf{a}, \boldsymbol{P}} \sum_{l=1}^L\frac{ R_l}{\omega_l} = B_g \sum_{l=1}^L \frac{\sum_{u=1}^{U} k_{l,u}}{\omega_l}   
=\max_{\mathbf{a},\boldsymbol{P}} B_g\sum_{l=1}^L 
\left(
\sum_{u=1}^U  \log_2\left(
1 - \frac{1.5 
\frac{
G^{\mathrm{Tx}} G^{\mathrm{Rx}}
\big|  \boldsymbol{H}_{l,u} \tilde{\boldsymbol{M}}_{u}\boldsymbol{P}_u  \big|^2 
}{
G^{\mathrm{Tx}} G^{\mathrm{Rx}}I_u + B_g n_0
}}{\ln\big(5\epsilon_{l,u}\big)}
\right)\right)/\omega_l,
\label{eq:p22}
\end{equation}
\noindent\rule{\textwidth}{0.4pt}
\end{figure*}

\begin{algorithm}[t]
\caption{Projected-Gradient (PG) Solver for $\mathcal{P}2$--2}
\label{alg:RHS}
\begin{algorithmic}[1]
\REQUIRE Channel $\boldsymbol H_{l,u}$, interference terms $I_u$, stepsizes $\eta_P,\eta_a$,
 stopping tolerance $\varepsilon$
\ENSURE Feeds power $\{\boldsymbol P_u\}_{u=1}^U$, holographic weights $\{a_{l,u,k}\}$
\STATE Initialize $\boldsymbol P_u^{(0)}$ such that
$\|\boldsymbol P_u^{(0)}\|^2=P_u^{\mathrm{Tx}}$, $P_u^{\mathrm{Tx}}=P^{\mathrm{Tx}}/U$
and initialize $a_{l,u,k}^{(0)}$, compute $\boldsymbol M_u^{(0)}$  for all $u$. Set  $i\leftarrow 0$.
\WHILE{$\big|J^{(i+1)}-J^{(i)}\big|>\varepsilon$}
    \STATE Compute gradients
    $\nabla_{\boldsymbol P_u^{\ast}} J^{(i)}$ and
    $\nabla_{a_{l,u,k}} J^{(i)}$ using Appendix~C. {Gradient ascent update:}
    $
    \boldsymbol P_u^{(i+1)}
    \leftarrow
    \boldsymbol P_u^{(i)} + \eta_P \nabla_{\boldsymbol P_u^{\ast}} J^{(i)},
   $
    $
    a_{l,u,k}^{(i+1)}
    \leftarrow
    a_{l,u,k}^{(i)} + \eta_a \nabla_{a_{l,u,k}} J^{(i)}.
   $
    \STATE Project $\boldsymbol P_u^{(i+1)}$ onto
        $\|\boldsymbol P_u\|^2=P_u^{\mathrm{Tx}}$;
        Project $\{a_{l,u,k}^{(i+1)}\}$ onto the simplex
        $\{a_{l,u,k}\ge0,\ \sum_{l,k}a_{l,u,k}=1\}$.
   \STATE Update $\boldsymbol M_u^{(i+1)}$ for all $u$. $i\leftarrow i+1$.
\ENDWHILE
\end{algorithmic}
\end{algorithm}

Assume the prefetching-rendering strategy (\(\boldsymbol{p, r}\)) has been determined.  We introduce Problem 1-2. 

\newtheorem*{problem1-2}{\textbf{Problem 1-2} \textnormal{(RHS holographic beamforming (given prefetching-rendering strategy))}} 
\begin{problem1-2}
The physical layer objective is to optimize the holographic pattern weights $\mathbf{a}$ and feeds power allocation vector $\boldsymbol{P}$ to maximize the total rate, which is based on the SINR across multiple subbands, expressed as eq. \eqref{eq:p22}, subject to constraints (C6) and (C7).\end{problem1-2}

Here, $\omega_l$ is derived from the optimal solution of $\mathcal{P}$1-1: $\omega_l = \sum_{i=1}^V \pi_{l, i} \left( 1 - p_{l, i}^{3\mathrm{D}} \right) \Big[ r_{l, i} \left( 1 - p_{l, i}^{2\mathrm{D}} \right) Q_{l, i}^{2\mathrm{D}} 
 + \left( 1 - r_{l, i} \right) Q_{l, i}^{3\mathrm{D}} \Big]$, which essentially represents the expected amount of critical data that still needs to be pulled from the wireless link under the current prefetching-rendering strategy. The higher this value, the more rate should be allocated at the physical layer. 

$\boldsymbol{H}_{l, u} \tilde{\boldsymbol{M}}_u \boldsymbol{P}_u$ is the signal transmitted from the RHS to HMD $l$,  $ \boldsymbol{H}_{l, u} \in \mathbb{C}^{1 \times N}, \tilde{\boldsymbol{M}}_u \in \mathbb{C}^{N \times K}, \boldsymbol{P}_u \in \mathbb{C}^{K \times 1}$, then:
\begin{equation}
\boldsymbol{H}_{l,u} \tilde{\boldsymbol{M}}_u \boldsymbol{P}_u
=
\sum_{n=1}^{N} h_{l,u,n}
\left(
\sum_{m=1}^{N} \Xi_{u,n,m}
\sum_{k=1}^{K} {\boldsymbol{M}_{u}}_{m}^{k} P_{u,k}
\right),
\end{equation}
where $\Xi_{u,n,m}$ denotes the $(n,m)$-th entry of $\boldsymbol{\Xi}_u$. $h_{l, u, n}=\frac{c}{4\pi f_u d_l}
\exp\!\left(-\frac{\kappa(f_u) d_l}{2}\right)$, ${\boldsymbol M_{u}}_n^k=\sqrt{\beta} m_{u,n}e^{-\alpha \|\boldsymbol{r}_n^k\|}e^{-j \boldsymbol{k}_s \boldsymbol{r}_n^{k}}$ with each $m_{u,n}
=\sum_{l=1}^{L}\sum_{k=1}^{K}
a_{l,u,k}
\frac{\operatorname{Re}\left[\Psi_{l,u}^{\text{intf}}\right]+1}{2}$.  Here, $h_{l, u, n}$, $\operatorname{Re}\left[\Psi_{l,u}^{\text{intf}}\right]$, $G^{\mathrm{Tx}} G^{\mathrm{Rx}}I_u + B_g n_0$ and $\ln\big(5\epsilon_{l,u}\big)$ are constants for HMD $l$ given $(\boldsymbol{p}_l^{3\mathrm{D}}, \boldsymbol{p}_l^{2\mathrm{D}}, \mathbf{r}_l)$.

$\mathcal{P}$1-2 is non-convex. A PG method (Algorithm \ref{alg:RHS}), with the gradients computation $\nabla_{\boldsymbol P_u} J$ and $\nabla_{a_{l,u,k}} J$ provided in Appendix C, is applied.

\section{Simulation}
\subsection{Simulation for the Homogenous Setting ($\mathcal{P}$1')}

\begin{figure*}[htbp]
    \centering
    \begin{subfigure}[t]{0.32\textwidth}
        \includegraphics[width=\textwidth]{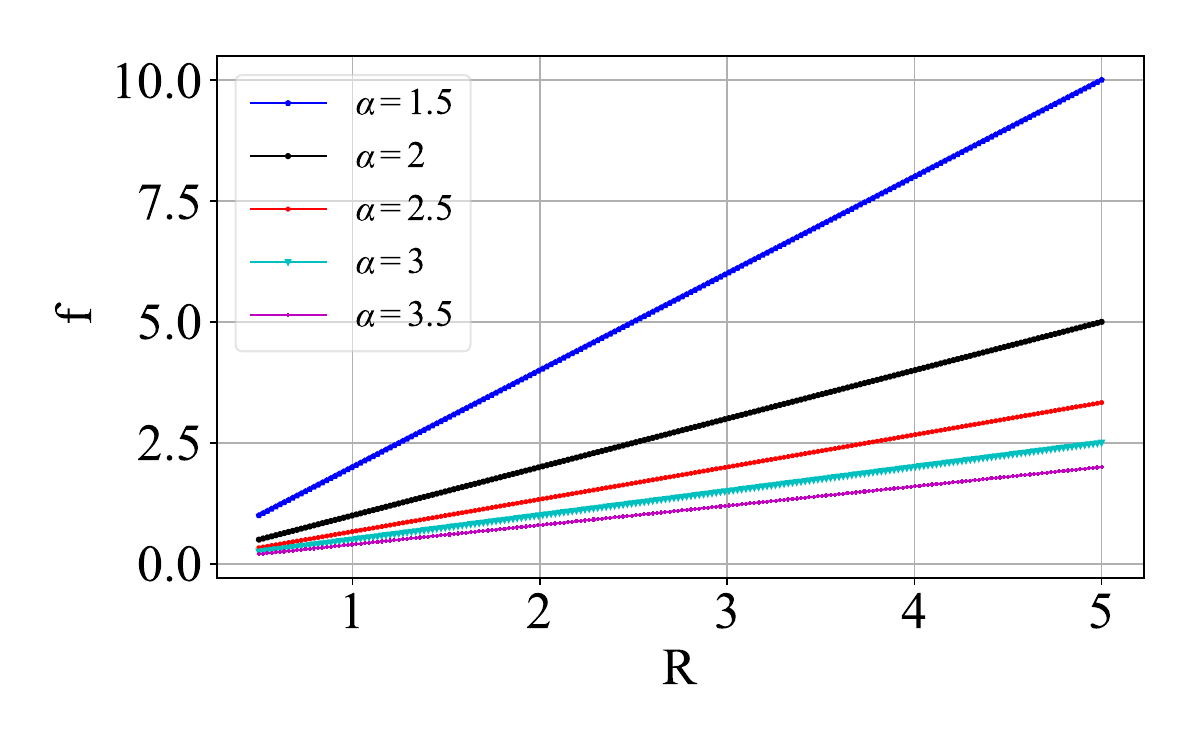}
        \caption{Zone boundary curves with different $\alpha$.}
        \label{fig:homo1}
    \end{subfigure}   
\hfill
    \begin{subfigure}[t]{0.32\textwidth}
        \includegraphics[width=\textwidth]{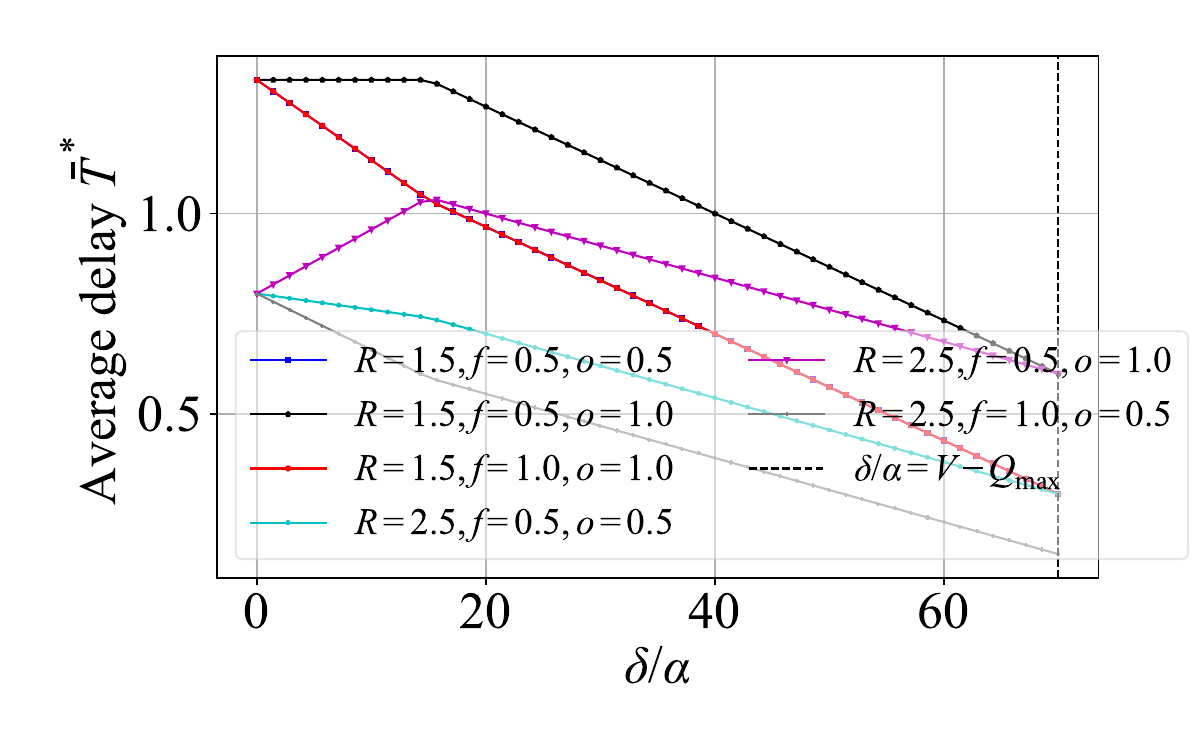}
        \caption{Zone 1 (remote 3D transmission priority zone), six $(R,f,o)$ combinations.}
        \label{fig:homo2}
    \end{subfigure}   
\hfill
    \begin{subfigure}[t]{0.32\textwidth}
        \includegraphics[width=\textwidth]{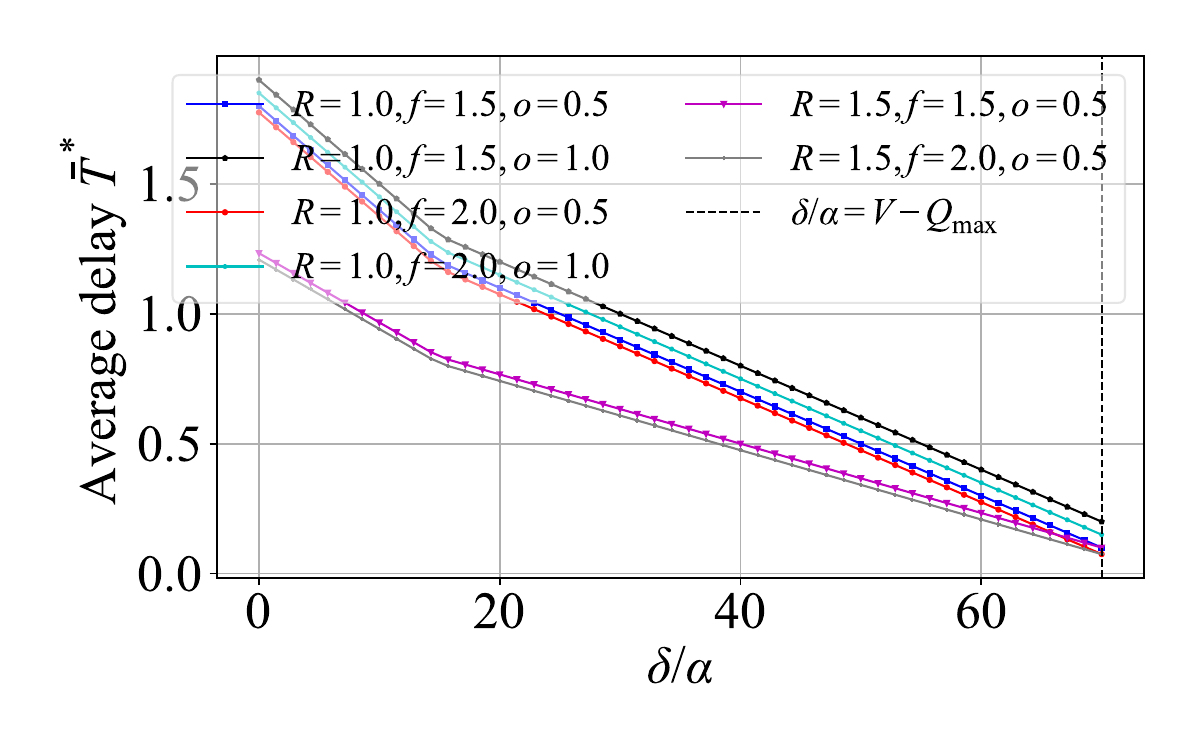}
        \caption{Zone 2 (on-device 2D rendering priority zone), six $(R,f,o)$ combinations.}
        \label{fig:homo3}
    \end{subfigure}   
\hfill
    \begin{subfigure}[t]{0.32\textwidth}
        \includegraphics[width=\textwidth]{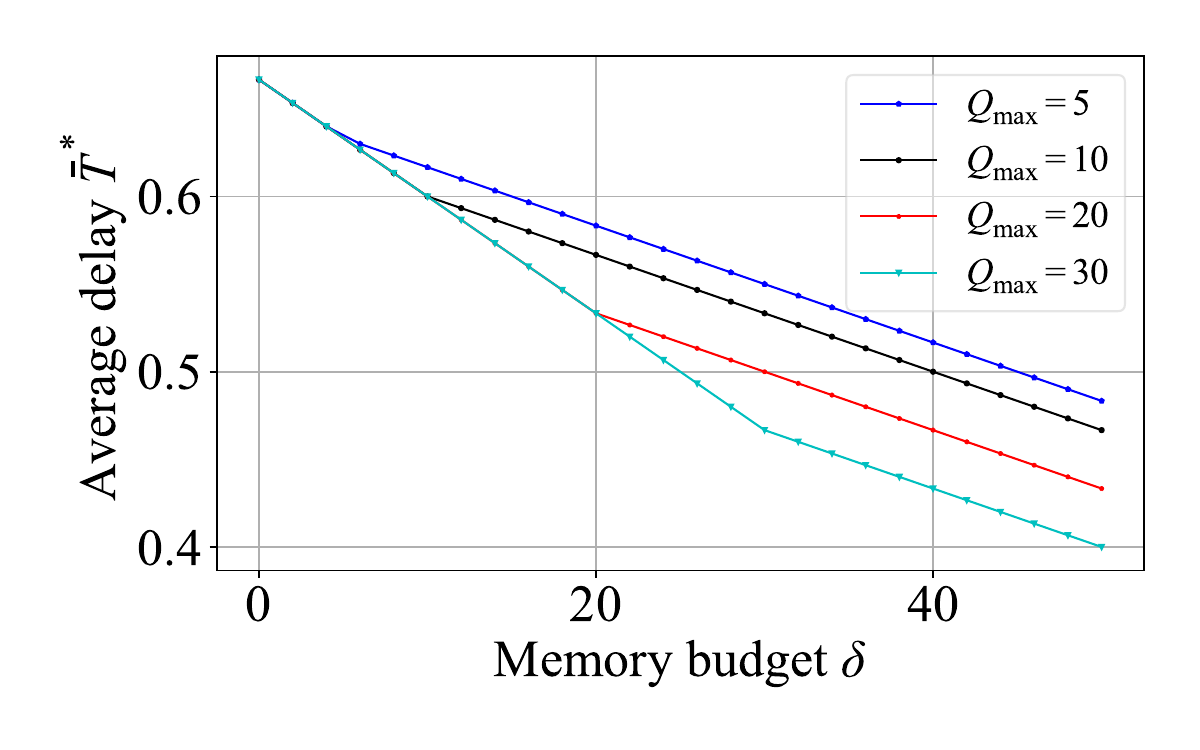}
        \caption{Zone 2 (on-device 2D rendering priority zone): $\bar T^*$ vs $\delta$.}
        \label{fig:homo4}
    \end{subfigure}   
\hfill
    \begin{subfigure}[t]{0.32\textwidth}
        \includegraphics[width=\textwidth]{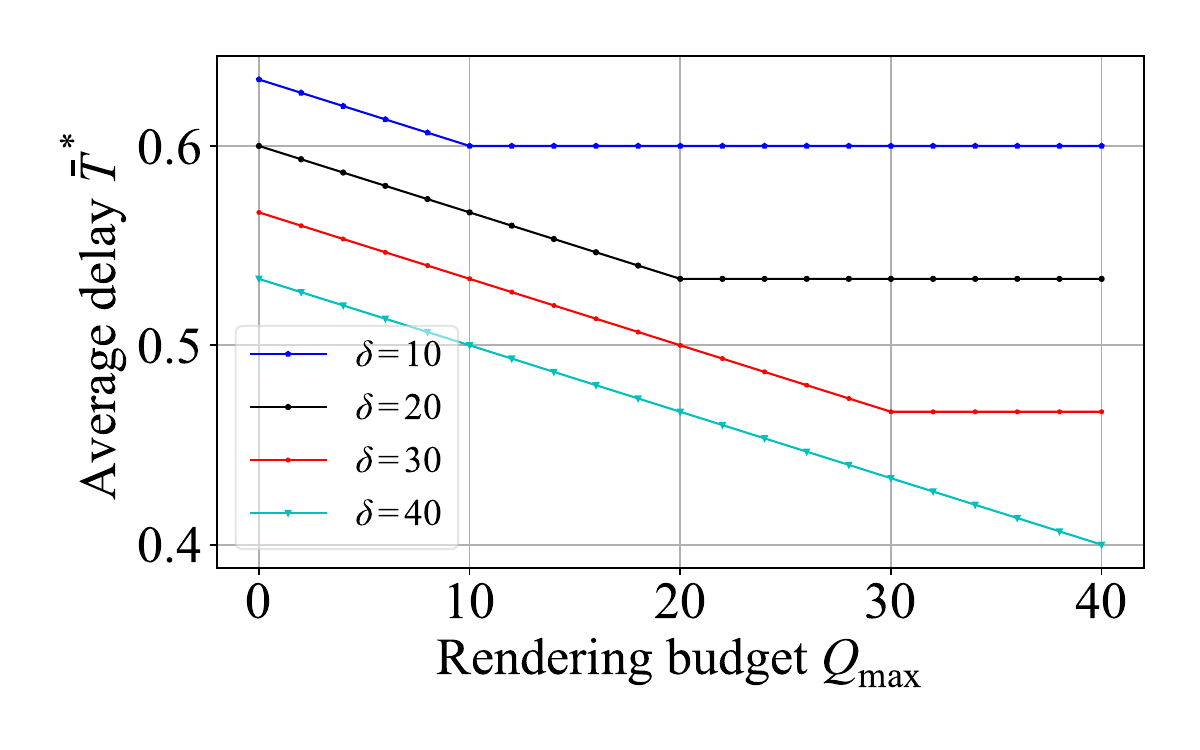}
        \caption{Zone 2 (on-device 2D rendering priority zone): $\bar T^*$ vs $Q_{\mathrm{max}}$.}
        \label{fig:homo5}
    \end{subfigure}   
\hfill
    \begin{subfigure}[t]{0.32\textwidth}
        \includegraphics[width=\textwidth]{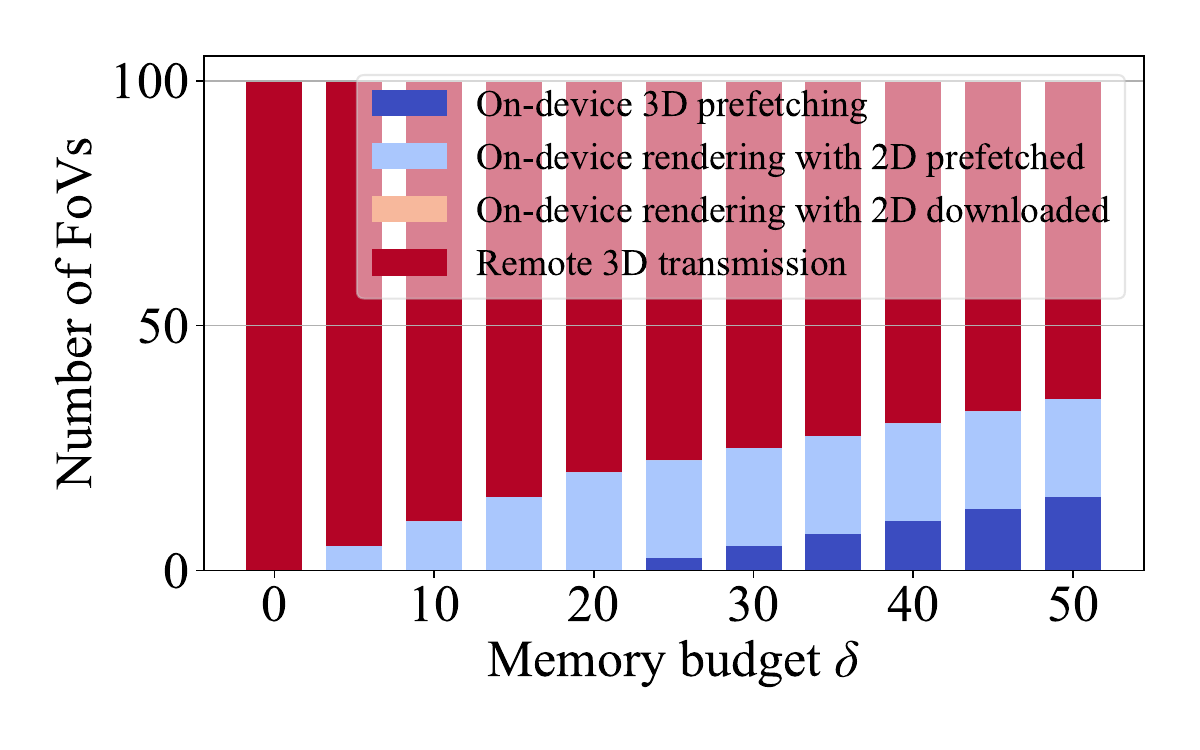}
        \caption{Zone 1 (remote 3D transmission priority zone), FoVs vs $\delta$.}
        \label{fig:homo6}
    \end{subfigure}   
    \caption{Simulation results for prefetching and rendering offloading in homogenous setting($\mathcal{P}$1').}
    \label{fig:homo}
\end{figure*}

\subsubsection{Setup}
We adopt a standard FoV modeling approach following~\cite{sun2019communications}: a $360^\circ$ panoramic video is projected into a 2D equirectangular format, partitioned into independently encoded H.264 tiles, and temporally segmented into 1-second chunks. To isolate the structural behavior of prefetching and rendering, we assume homogeneous FoVs with normalized size $Q^{2\mathrm{D}}=1$ (about 3 Mbits), while 3D FoVs satisfy $Q^{3\mathrm{D}}=\alpha Q^{2\mathrm{D}}$ with $\alpha=2$ (about 6 Mbits) owing to stereoscopic rendering. Unless specified, default $R=1$ Gbps. A total of $V=100$ FoVs is used to represent diverse but tractable view directions. The maximum tolerable latency is $\tau=20$ ms \cite{du2025qoe}, the local CPU uses switched-capacitance parameter $\zeta=10^{-27}$ J·s·cycle$^{-3}$, rendering cost $o=1$ cycles/bit, and CPU frequency $f\in[1.5,5]$ GHz \cite{deng2022fov} under power budget $\bar P=5$ W.

\subsubsection{Results and Discussion}
Fig.~\ref{fig:homo1} shows the $(R,f)$-plane boundary separating the ``remote 3D transmission priority'' and ``on-device rendering priority'' zones. The boundary $o/f = (\alpha-1)/R$ steepens with increasing $\alpha$, indicating that heavier 3D data requires higher CPU frequency or stronger wireless links for ``remote 3D transmission'' to remain preferable. 

Fig.~\ref{fig:homo2} plots the average delay $\bar{T}^*$ as the normalized memory size $\delta/\alpha$ grows in the ``remote 3D transmission priority'' zone. All curves decay since larger memory reduces dependence on ``remote 3D transmission''. The decay rate varies with $(R,f,o)$: larger $R$ yields milder reduction, while higher $o$ or lower $f$ yields steeper curves. The vertical dashed line at $\delta/\alpha = V - Q_{\mathrm{max}}$ marks the point where all FoVs can be served locally, and the kink near this boundary matches the structure in Proposition \ref{theorem1}.

Fig.~\ref{fig:homo3} examines the case where ``on-device rendering'' is strictly superior to ``remote 3D transmission''. Compared with Fig.~\ref{fig:homo2}, the decay is sharper—especially for $\delta/\alpha \le 20$—reflecting the high marginal benefit of prefetching when ``local rendering'' dominates. CPU capability and content complexity strongly shape the curve ordering, with higher $f$ or smaller $o$ greatly reducing delay. The vertical dashed line again aligns with the boundary from Proposition \ref{theorem2}.

Fig.~\ref{fig:homo4} shows how $\bar{T}^*$ decreases with memory size $\delta$ in the ``on-device rendering priority'' zone. Prefetching remains the dominant performance enhancer: more prefetched 2D FoVs increase reliance on the fast ``on-device rendering with 2D prefetched'' path, reducing ``remote 3D transmission''. While larger $Q_{\mathrm{max}}$ improves baseline delay, the marginal prefetching benefit is more pronounced when $Q_{\mathrm{max}}$ is small, underscoring their complementarity.

Fig.~\ref{fig:homo5} illustrates the delay reduction as $Q_{\mathrm{max}}$ increases under different $\delta$. As $Q_{\mathrm{max}}$ increases, delay drops most when memory is scarce (small $\delta$), since stronger rendering cuts dependence on ``remote 3D transmission''; when memory is abundant (large $\delta$), the curves flatten and performance becomes prefetching-limited.

Fig.~\ref{fig:homo6} depicts how FoVs are allocated across service modes as memory budget $\delta$ increases in the ``remote 3D transmission zone''. As $\delta$ grows, FoVs reduce from ``remote 3D transmission'' to ``on-device 3D prefetching'' and ``rendering'': ``remote transmission'' dominates at small $\delta$, ``on-device rendering'' appears at moderate $\delta$, and at large $\delta$ ``remote 3D'' is largely replaced by ``on-device 3D prefetching'' and ``on-device 2D rendering'', which is exactly matching the optimal structure in Proposition \ref{theorem1}.

\begin{figure*}[htbp]
    \centering
    \begin{subfigure}[t]{0.32\textwidth}
        \includegraphics[width=\textwidth]{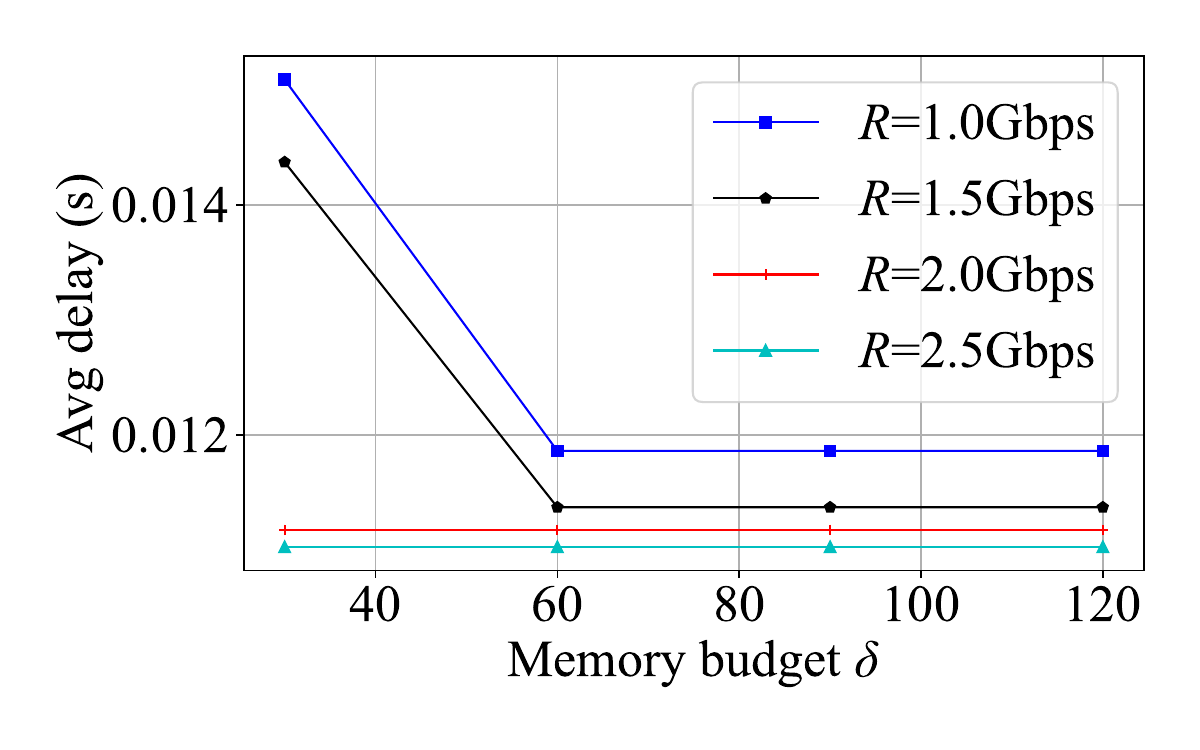}
        \caption{ $\bar T^*$ vs $\delta$.}
        \label{fig:heter1}
    \end{subfigure}   
\hfill
    \begin{subfigure}[t]{0.32\textwidth}
        \includegraphics[width=\textwidth]{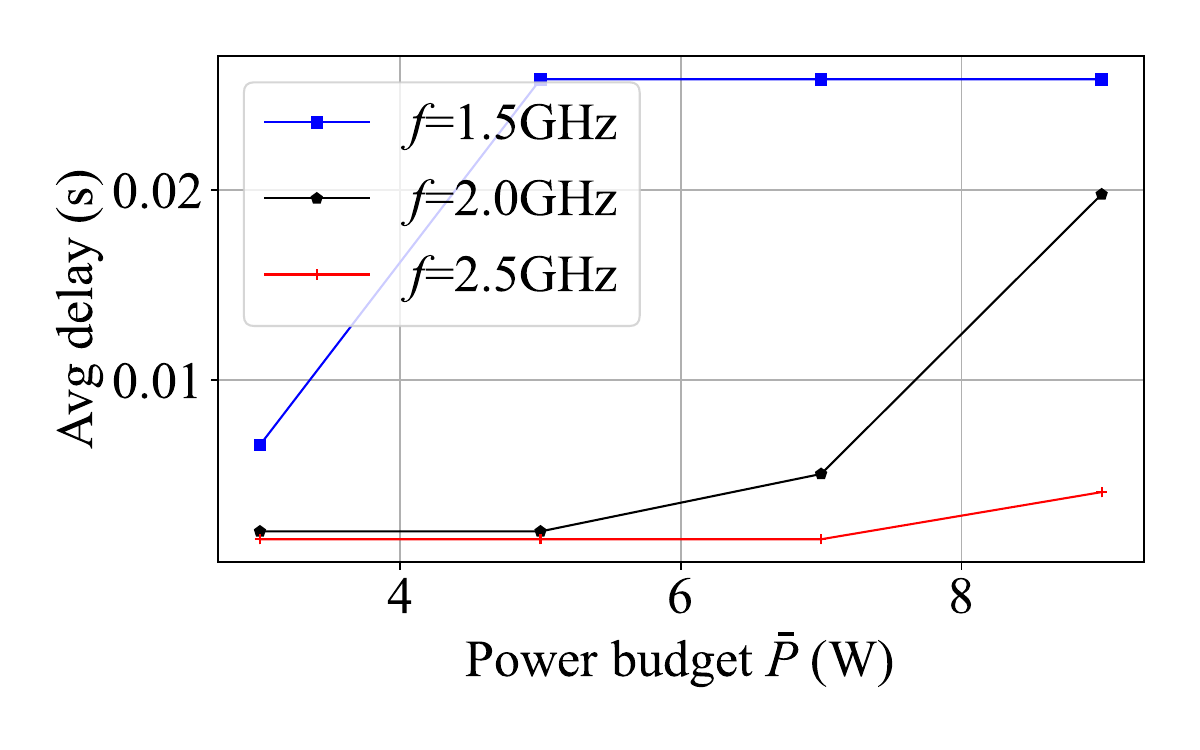}
        \caption{ $\bar T^*$ vs $\bar{P}_l$.}
        \label{fig:heter2}
    \end{subfigure}   
\hfill
    \begin{subfigure}[t]{0.32\textwidth}
        \includegraphics[width=\textwidth]{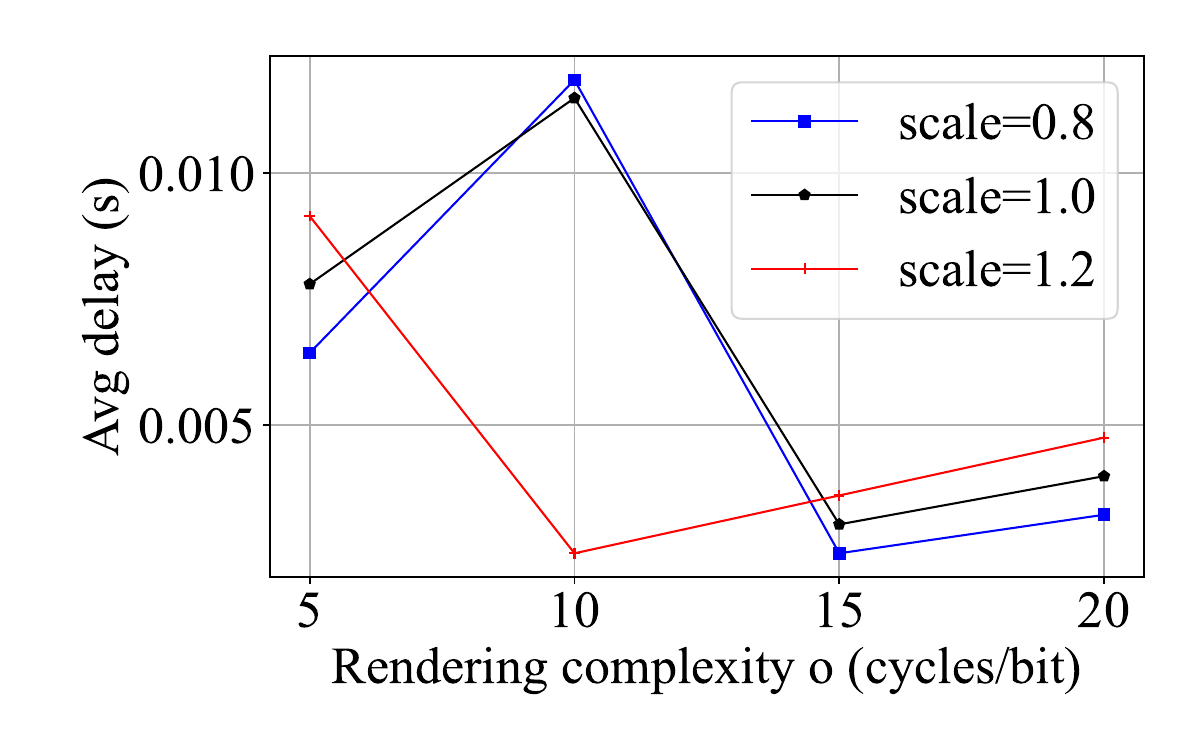}
        \caption{ $\bar T^*$ vs $o$.}
        \label{fig:heter3}
    \end{subfigure}   
\hfill
    \begin{subfigure}[t]{0.32\textwidth}
        \includegraphics[width=\textwidth]{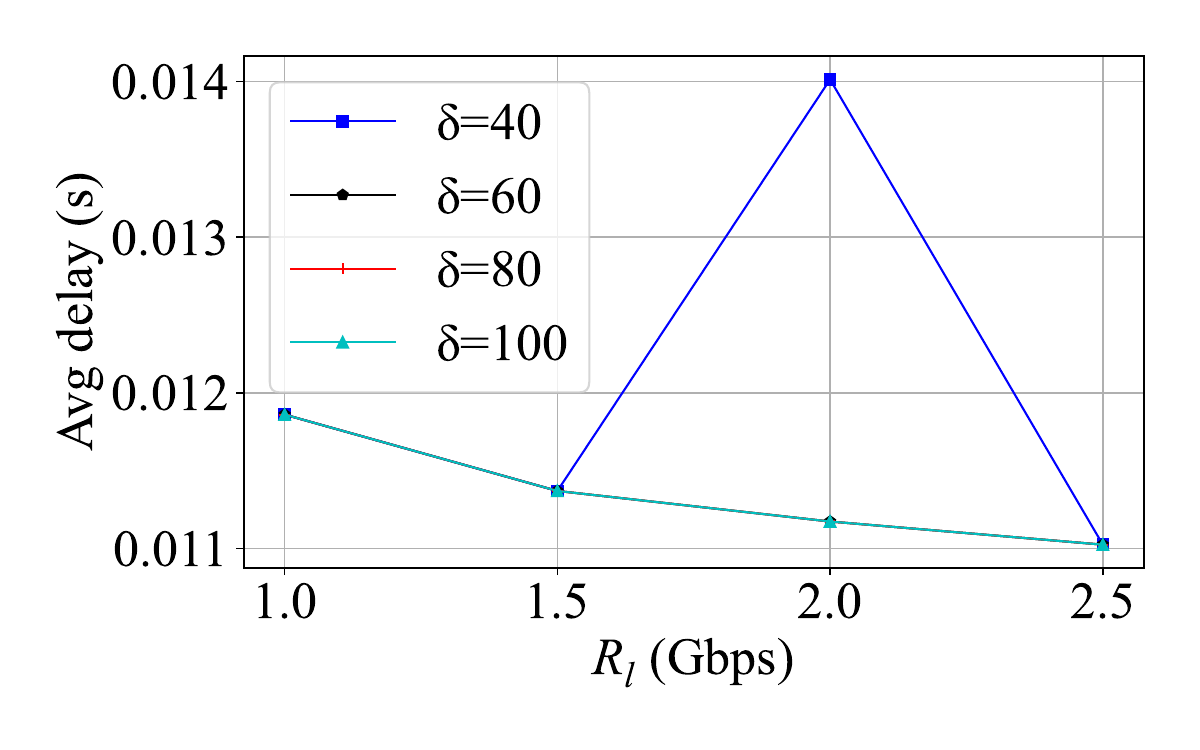}
        \caption{ $\bar T^*$ vs $R_l$.}
        \label{fig:heter4}
    \end{subfigure}   
\hfill
    \begin{subfigure}[t]{0.32\textwidth}
        \includegraphics[width=\textwidth]{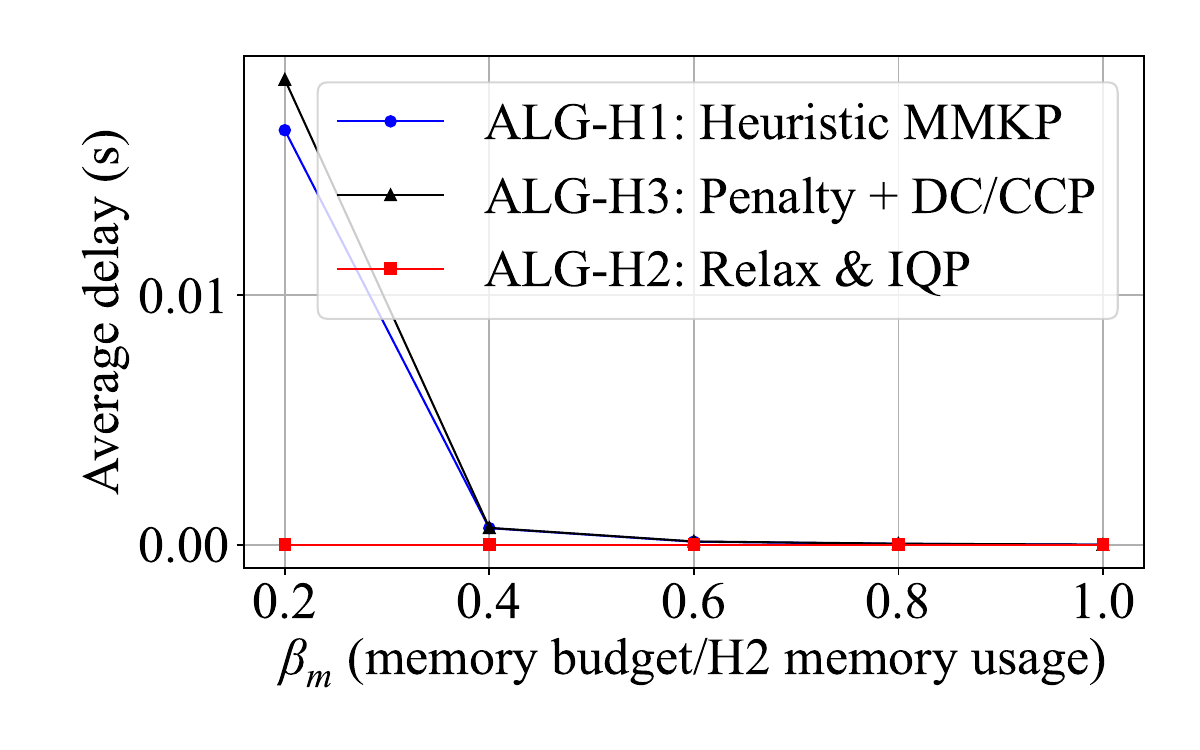}
        \caption{ $\bar T^*$ vs $\beta_m$.}
        \label{fig:heter5}
    \end{subfigure}   
\hfill
    \begin{subfigure}[t]{0.32\textwidth}
        \includegraphics[width=\textwidth]{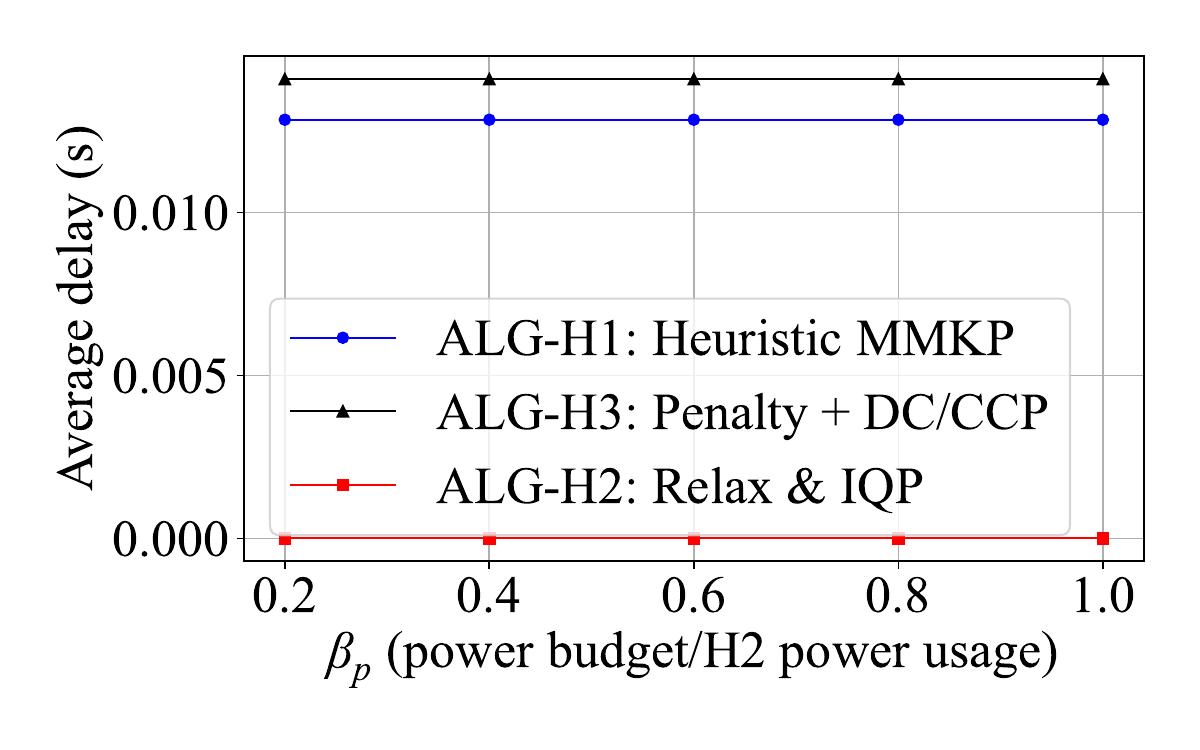}
        \caption{ $\bar T^*$ vs $\beta_p$.}
        \label{fig:heter6}
    \end{subfigure}   
    \caption{Simulation results for prefetching and rendering offloading in heterogeneous setting ($\mathcal{P}$1-1).}
    \label{fig:heter}
\end{figure*}

\subsection{Simulation of Prefetching and Rendering Offloading for the Heterogeneous Setting ($\mathcal{P}$1-1)}
\subsubsection{Setup}
We consider a multi-user VR system with \(L=4\) HMDs and \(V=100\) FoVs per HMD. For each user \(l\) and FoV \(i\), heterogeneous FoV sizes follow  
$Q_{l,i}^{2\mathrm{D}} \sim \mathcal{U}[1,4]~\text{Mbit}, Q_{l,i}^{3\mathrm{D}} \sim \mathcal{U}[2,8]~\text{Mbit},$
and the rendering complexity is drawn as  
$o \sim \mathcal{U}[5,15]~\text{cycles/bit}.$
FoV request probabilities follow a Zipf distribution with exponent \(\gamma=1.2\). 
The delay threshold is \(\tau = 20\) ms. Local CPU frequencies are heterogeneous, \(f_l \in \{1.5,2.0,2.5,3\}\) GHz, with switched capacitance \(\zeta = 10^{-27}\)  J·s·cycle$^{-3}$. The long-term power budgets \(\bar P_l\) vary across experiments. Downlink rates are constant during each run and take values  
$R \in \{1.0,1.5,2.0,2.5\}~\text{Gbps},$
allowing evaluation under heterogeneous channel conditions.

\begin{figure*}[htbp]
    \centering
    \begin{subfigure}[t]{0.32\textwidth}
        \includegraphics[width=\textwidth]{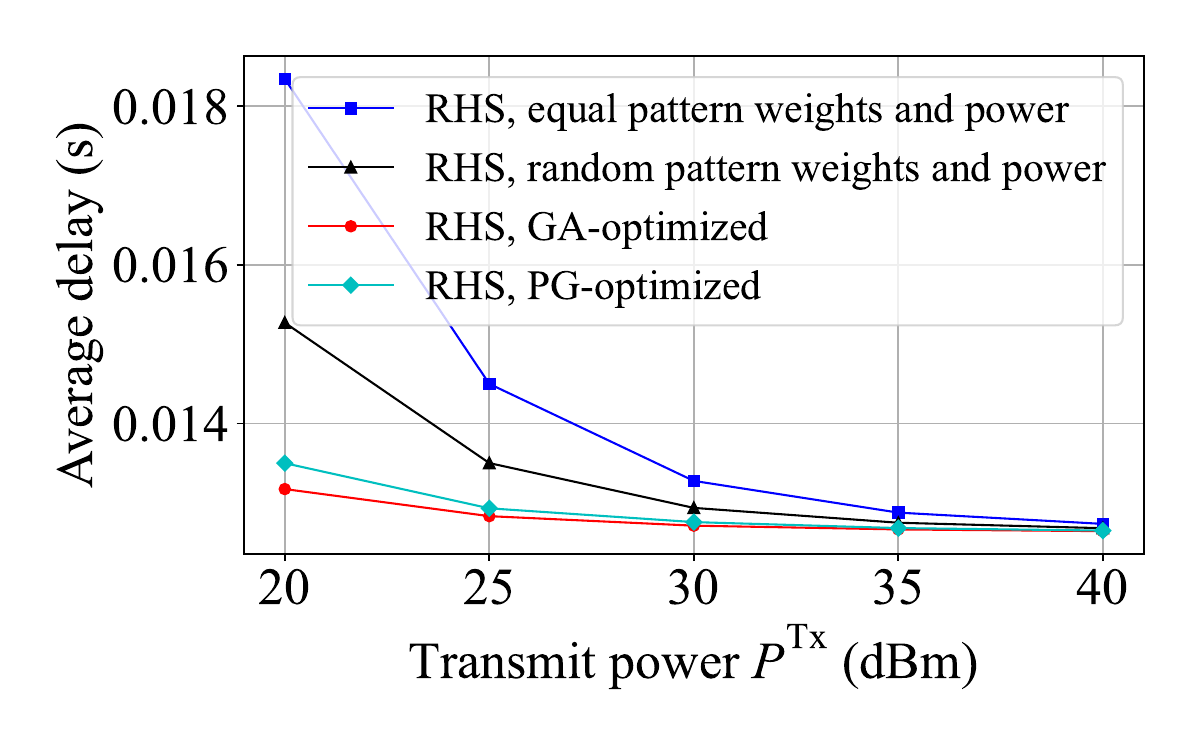}
        \caption{$\bar T^*$ vs $P_\mathrm{Tx}$.}
        \label{fig:rhs1}
    \end{subfigure}   
\hfill
    \begin{subfigure}[t]{0.32\textwidth}
        \includegraphics[width=\textwidth]{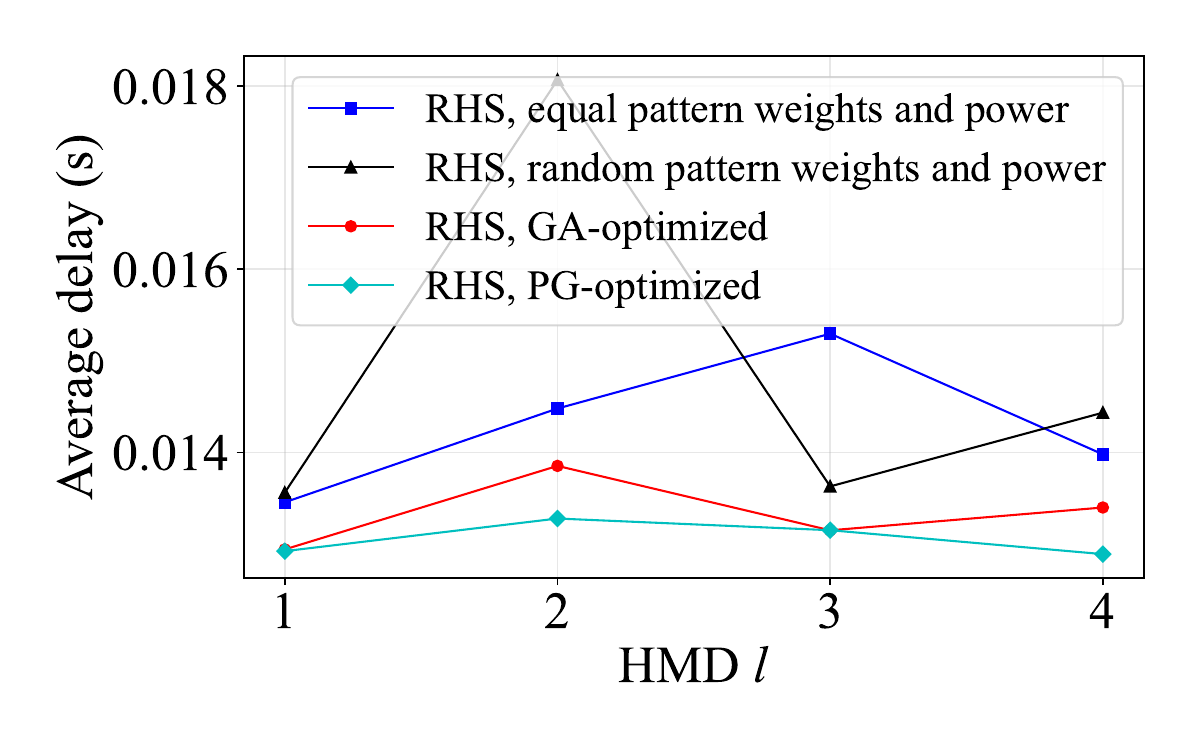}
        \caption{Comparison of HMDs.}
        \label{fig:rhs2}
    \end{subfigure}   
\hfill
    \begin{subfigure}[t]{0.32\textwidth}
        \includegraphics[width=\textwidth]{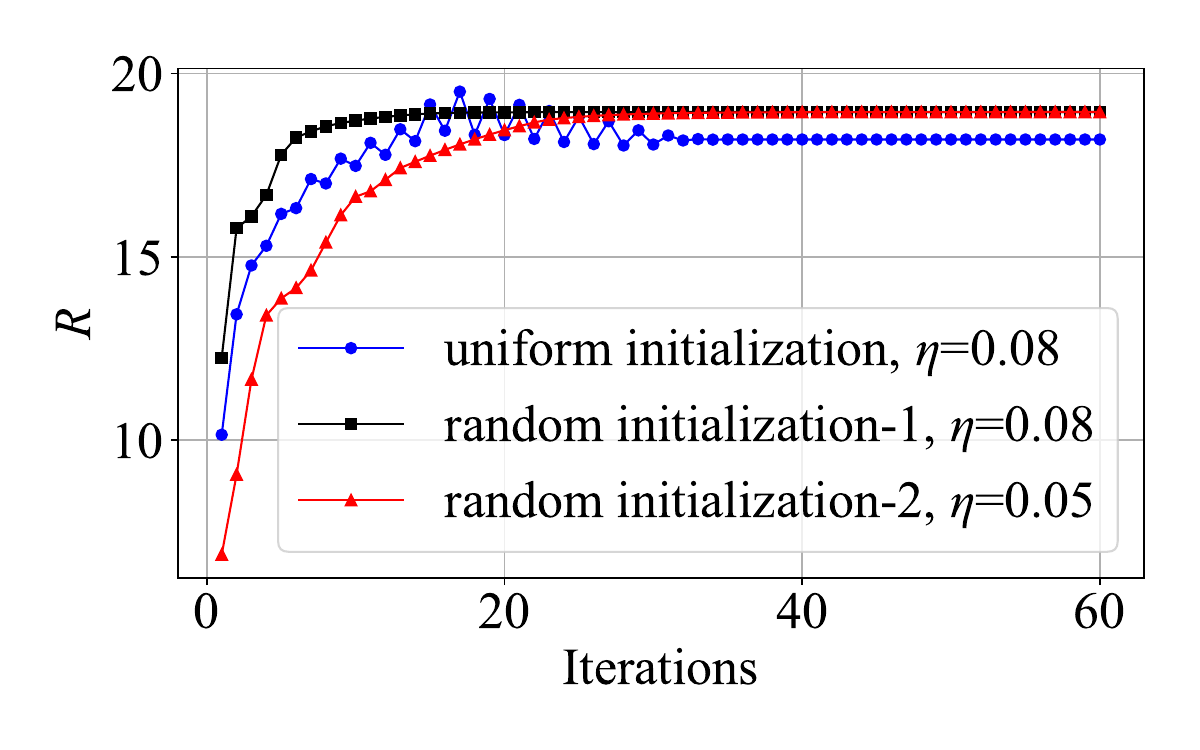} 
        \caption{Convergence of PG-optimizer.}
        \label{fig:rhs3}
    \end{subfigure}     
    \caption{Simulation results for RHS holographic pattern weights adjustment and feeds power allocation ($\mathcal{P}$1-2).}
    \label{fig:rhs}
\end{figure*}

\subsubsection{Results and Discussion}
Fig.~\ref{fig:heter1}–\ref{fig:heter6} jointly illustrate the impact of key system parameters and the effectiveness of the three algorithms under heterogeneous FoVs. Figs.~\ref{fig:heter1}–\ref{fig:heter4} show that enlarging the memory budget $\delta$, the power budget $\bar P_l$, the rendering capability (larger $f_l$ or smaller $o$), or the wireless rate $R_l$ all reduce the optimal average delay $\bar T^*$, with the gains being most pronounced in the corresponding bottleneck regimes (small $R_l$, small $\bar P_l$, large $o$, or small $\delta$). Figs.~\ref{fig:heter5} and~\ref{fig:heter6} further compare ALG-H1, ALG-H3 and ALG-H2 as the memory and power budget ratios $\beta_m$ and $\beta_p$ vary, where $\beta_m$ (resp. $\beta_p$) denotes the available memory (resp. power) normalized by the usage of the relaxed lower-bound solution ALG-H2. ALG-H2 yields a constant performance lower bound, while ALG-H1 exhibits the largest delay, especially when $\beta_m$ or $\beta_p$ is small. ALG-H3 consistently outperforms ALG-H1 and rapidly approaches the ALG-H2 bound as $\beta_m$ or $\beta_p$ increases, confirming that the proposed relaxation–and–DC/CCP refinement effectively exploits additional cache and power resources to approximate the ideal  solution under  real constraints.

\subsection{Simulation of RHS Beamforming for the Heterogeneous Setting ($\mathcal{P}$1-2)}
\subsubsection{Setup}
We consider a single THz BS equipped with an RHS that has $N=N_xN_y=256$ elements with $N_x=N_y=16$ placed on a $\sqrt{N}\times\sqrt{N}$ UPA with $K=3$ feed sources.  $Z_O=50~\Omega$, and $\mathbf{Z}_O=Z_O\mathbf{I}_N$, $\mathbf{I}_N\in\mathbb{C}^{256\times256}$.  $l_d=\lambda_c/4=0.25~\mathrm{mm}$, and
$Z_{A,u}=80\left(\frac{l_d}{\lambda_u}\right)^2-\mathrm{j} \frac{120}{\pi} \ln \left(\frac{\lambda_u}{2 \pi l_d}\right)$, $\mathbf{Z}_{A,u}=Z_{A,u}\mathbf{I}_N$.
The RHS phase center is located at $\mathbf{p}_{\mathrm{R}}=[0,0,2]^{\top} \mathrm{m}$ on the wall of an indoor room of size $6 \mathrm{~m} \times 6 \mathrm{~m} \times 3 \mathrm{~m}$. Each of the $L=4$ HMD users is placed at $\mathbf{p}_l=\left[x_l, y_l, 1.6\right]^{\top} \mathrm{m}$, where  $(x_l, y_l)$ are drawn uniformly from $[1,5] \times[1,5]$. The molecular absorption coefficient $\kappa(f_u)=3\times 10^{-3}\ \mathrm{m}^{-1}$; transmit power is $P^{\mathrm{Tx}}=30\;\mathrm{dBm}$ and is equally shared across subbands, or equivalently the SNR is swept from -10 to 20 dB by scaling $P^{\mathrm{Tx}}$. The waveform spectrum $G_v(f)$ satisfies $\int |G_v(f)|^2 df = 1$. The index set $\mathcal{Q}$ contains the dominant leakage components from adjacent subbands and is set to $\mathcal{Q}=\{1\}$, with leakage coefficient $\alpha_{v,q}=0.05$ for all $v\neq u$. Transmit and receive antenna gains are set to $G^{\mathrm{Tx}}=30\;\mathrm{dBi}$ and $G^{\mathrm{Rx}}=10\;\mathrm{dBi}$, the receiver noise spectral density is $n_0=-174 \;\mathrm{dBm} / \mathrm{Hz}$. The RHS employs structural attenuation coefficient $\alpha=0.1 \mathrm{~m}^{-1}$, radiation efficiency $\beta=0.8$, and a surface-wave vector aligned with the $x$-axis. Holographic pattern weights $a_{l, u, k}$ satisfying $\sum_{l, k} a_{l, u, k}=1$. $\mathcal{P}$1-2 is solved via a projected gradient ascent method with fixed step size $\eta_{\mathrm{P}}=\eta_a=0.05$, maximum number of iterations $N_{\text {iter }}=200$, and stopping tolerance $\varepsilon=10^{-4}$ on the relative improvement of the weighted sum rate, where after each gradient update the vectors $\{P_n\}_{n=1}^{N}$ and $\left\{a_{l, u, k}\right\}_{l=1,k=1}^{L,K}$ are projected onto their corresponding probability simplices.

\subsubsection{Results and Discussion}
Fig. \ref{fig:rhs} shows that joint optimization of RHS beamforming (holographic pattern weights optimization and feeds power allocation with genetic algorithm(GA)- and PG-based optimizer) significantly reduces the average delay compared with equal or random pattern weights and power, especially at moderate transmit powers. Across different HMDs, the optimized schemes not only achieve lower delay but also provide much more balanced performance, avoiding the severe degradation observed for some users under equal-weight designs. The convergence curves further indicate that the proposed PG-optimizer rapidly converges within a few tens of iterations and is relatively robust to different initializations and step sizes.

\section{Conclusion}
This paper presents a framework for jointly optimizing prefetching, rendering, RHS-enabled THz beamforming by adjusting holographic pattern weights, and feeds power allocation in multi-user wireless VR systems. Additionally, mutual coupling effects between adjacent elements are incorporated into the model, ensuring more accurate optimization of the beamforming process. For homogeneous FoVs, we derive closed-form structural policies that reveal two operating zones: remote-3D-transmission–priority zone and on-device-rendering–priorit zone, and show how memory, computing capacity, and content complexity determine the optimal strategy. For heterogeneous FoVs, we leverage the timescale separation between slow prefetching/rendering decisions and fast RHS re-beamforming under user mobility, leading to a three-stage solution. The DC/CCP-based approach achieves near-optimal delay under memory and power constraints. We further formulate a smooth RHS beamforming  problem and develop an efficient projected-gradient solver that consistently outperforms unoptimized baselines. Simulations across all scenarios show that the framework substantially reduces VR latency. Future work includes exploring more heuristics-based solvers, reinforcement-learning-based approaches for real-time mobility handling, and fully joint optimization without problem decomposition.

\bibliographystyle{Bibliography/IEEEtranTIE}
\bibliography{Bibliography/IEEEabrv,Bibliography/mybibfile.bib}\ 

\newpage

\section*{Appendix A: Proof of Proposition \ref{theorem1}}
\begin{proof}
Consider the following conditions. In the remote 3D transmission priority zone,  $T^{\langle 3 \rangle}
  \geq T^{\langle 4 \rangle}
 $, the following holds:  
for any FoV that does not prefetch 2D, if local rendering is used (path 3), its delay $T^{\langle 3 \rangle}
 $ is not smaller than the delay of directly using remote 3D transmission (path 4), which is $T^{\langle 4 \rangle}
 $.  Therefore, in the optimal solution, the situation where ``a FoV without 2D prefetching is assigned to local rendering'' will never occur. This means that at the optimal point, we can set $r \leq p^{2\mathrm{D}}$, which is equivalent to $q = \min \{p^{2\mathrm{D}}, r\} = r$.
Thus, all the FoVs rendered locally must have 2D FoV prefetched in memory, corresponding to path 2. Otherwise, the ``on-device rendering without 2D prefetched'' FoVs should be changed to ``remote 3D transmission'', which will not increase the total delay but will reduce it.

 Under the constraint $r \leq p^{2\mathrm{D}}$, the optimal solution is $r^* = p^{2\mathrm{D}^*}$. In the case where $q = r$, the average delay is:
\begin{equation}
\bar{T} = \frac{1}{V} \left[ r T^{\langle 2 \rangle}
  + \left(V - p^{3\mathrm{D}} - r \right) T^{\langle 4 \rangle}
  \right].
\end{equation}
The coefficient with respect to $r$ is:
$
\frac{\partial \bar{T}}{\partial r} = \frac{T^{\langle 2 \rangle}
  - T^{\langle 4 \rangle}
 }{V}.
$
In a typical VR setup (where rendering time is shorter than downloading 3D, i.e., $T^{\langle 2 \rangle}
  < T^{\langle 4 \rangle}
 $), $\bar{T}$ decreases strictly as $r$ increases.  
Since $r \leq p^{2\mathrm{D}}$, the optimal value of $r$ within the feasible zone must be $r^* = p^{2\mathrm{D}}$.  
Thus, we can rewrite $\bar{T}$ as a function of $p^{2\mathrm{D}}$ alone.

Substitute the memory constraint to obtain a linear objective function in terms of $p^{2\mathrm{D}}$. At this point, since $p^{3\mathrm{D}} = \frac{\delta - p^{2\mathrm{D}}}{\alpha}$, we have:
\begin{equation}
\begin{aligned}
\bar{T} &= \frac{1}{V} \left[ p^{2\mathrm{D}} T^{\langle 2 \rangle}
  + \left(V - \frac{\delta - p^{2\mathrm{D}}}{\alpha} - p^{2\mathrm{D}} \right) T^{\langle 4 \rangle}
  \right] \\
&= T^{\langle 4 \rangle}
  - \frac{T^{\langle 4 \rangle}
 }{V} \cdot \frac{\delta}{\alpha} + \frac{p^{2\mathrm{D}}}{V} \left(T^{\langle 2 \rangle}
  - T^{\langle 4 \rangle}
  + \frac{T^{\langle 4 \rangle}
 }{\alpha}\right).
\end{aligned}
\end{equation}
Note that $T^{\langle 4 \rangle}
  = \frac{\alpha Q^{2\mathrm{D}}}{R}$, hence
\begin{equation}
T^{\langle 4 \rangle}
  - \frac{T^{\langle 4 \rangle}
 }{\alpha} = T^{\langle 4 \rangle}
  \left(1 - \frac{1}{\alpha}\right) = \frac{(\alpha - 1) Q^{2\mathrm{D}}}{R}.
\end{equation}

Typically, in VR settings, we can assume
$
T^{\langle 2 \rangle}
  < T^{\langle 4 \rangle}
  \left(1 - \frac{1}{\alpha}\right),
$
i.e., on-device rendering with 2D prefetched provides more delay reduction than ``allocating part of the memory for 3D''. Thus, the coefficient in the parentheses is negative, and $\bar{T}$ decreases monotonically with respect to $p^{2\mathrm{D}}$. 
Therefore, under the constraints
$
p^{2\mathrm{D}} \leq \delta,  p^{2\mathrm{D}} \leq Q_{\mathrm{max}},
$
the optimal solution is
$
p^{2\mathrm{D}^*} = \min \left\{\delta, Q_{\mathrm{max}}\right\},  r^* = p^{2\mathrm{D}^*}.
$
Since $p^{3\mathrm{D}^*} = \frac{\delta - p^{2\mathrm{D}^*}}{\alpha}$, this naturally gives the amount of remaining memory used for 3D prefetching.  
Substituting into the expression for $\bar{T}$ yields the closed-form minimum average delay as stated in the proposition.

\end{proof}

\section*{Appendix B: Proof of Proposition \ref{theorem2}}
\begin{proof}
In the on-device 2D rendering priority zone, where $T^{\langle 3 \rangle}
  < T^{\langle 4 \rangle}
 $, local rendering is faster than remote 3D transmission even without 2D prefetched in its memory. Therefore, local rendering always helps reduce the delay, whether or not 2D FoV is prefetched. Thus, the optimal solution must have $r^* = Q_{\mathrm{max}}$, meaning that the full local memory budget is used. Moreover, if $p^{2\mathrm{D}} \leq r$: some local memory do not have 2D prefetched and follow path 3; if $p^{2\mathrm{D}} \geq r$: all local FoVs have 2D prefetched and follow path 2. For these two cases, we can express $\bar{T}$ separately. We can prove that: in the zone where $p^{2\mathrm{D}} \leq r$, $\bar{T}$ decreases monotonically with $p^{2\mathrm{D}}$ and with $r$. Hence, the optimal solution is $p^{2\mathrm{D}^*} = \min \{r, \delta\}$ and $r^* = Q_{\mathrm{max}}$; in the zone where $p^{2\mathrm{D}} \geq r$, the optimal point is at the boundary $p^{2\mathrm{D}^*} = r^*$, and $r^* = Q_{\mathrm{max}}$. Both zones achieve the same value at the boundary $p^{2\mathrm{D}} = r$. Thus, the global optimal solution is $r^* = Q_{\mathrm{max}}$,  $p^{2\mathrm{D}^*} = \min \left\{\delta, Q_{\mathrm{max}}\right\}$,  $p^{3\mathrm{D}^*} = \frac{\delta - p^{2\mathrm{D}^*}}{\alpha}$.

Substituting the above optimal strategy into:
\begin{equation}
\bar{T} = \frac{1}{V} \left[ q T^{\langle 2 \rangle}
  + (r - q) T^{\langle 3 \rangle}
  + \left( V - p^{3\mathrm{D}} - r \right) T^{\langle 4 \rangle}
  \right], 
\end{equation}
where $q = \min \left\{ p^{2\mathrm{D}}, r \right\}$. We can group and compute the delays for FoVs with 2D prefetched in local memory (path 2), FoVs without 2D downloaded from remove MECS (path 3), and remote 3D transmission (path 4). We get the following results:
\begin{enumerate}
\renewcommand{\labelenumi}{\roman{enumi}:}
\item  $p^{2\mathrm{D}^*}$ FoVs following path 2 with delay $T^{\langle 2 \rangle}
 $.
\item  $Q_{\mathrm{max}} - p^{2\mathrm{D}^*}$ FoVs following path 3 with delay $T^{\langle 3 \rangle}
 $.
\item  $V - p^{3\mathrm{D}*} - Q_{\mathrm{max}}$ FoVs following path 4 with delay $T^{\langle 4 \rangle}
 $.
\end{enumerate}

After expanding and simplifying, we obtain the closed-form expression for $\bar{T}^*$ as stated in the proposition.  
\end{proof}

\section*{Appendix C: Gradient Computation for  $\mathcal{P}2$--2}
Define the objective function as
\begin{equation}
J(\mathbf{ a},\boldsymbol P)
=B_g \sum_{l=1}^{L}\frac{1}{\omega_l}
\sum_{u=1}^{U}
\log_2\!\left(
1-B
\big|
\boldsymbol H_{l,u}\tilde{\boldsymbol M}_u\boldsymbol P_u
\big|^2
\right),
\end{equation}
where $B=\frac{1.5}{\ln(5\epsilon_{l,u})}
\frac{G^{\mathrm{Tx}}G^{\mathrm{Rx}}}
{G^{\mathrm{Tx}}G^{\mathrm{Rx}}I_u+B_g n_0}
$ is a constant for $(l,u)$ and $\tilde{\boldsymbol M}_u\triangleq \boldsymbol \Xi_u\boldsymbol M_u$.
Using Wirtinger calculus, the gradient of
$
\big|
\boldsymbol H_{l,u}\tilde{\boldsymbol M}_u\boldsymbol P_u
\big|^2
$
with respect to $\boldsymbol P_u^{\ast}$ is given by
$
\frac{\partial}{\partial \boldsymbol P_u^{\ast}}
\big|
\boldsymbol H_{l,u}\tilde{\boldsymbol M}_u\boldsymbol P_u
\big|^2
=
\tilde{\boldsymbol M}_u^{H}\boldsymbol H_{l,u}^{H}
\big(
\boldsymbol H_{l,u}\tilde{\boldsymbol M}_u\boldsymbol P_u
\big).
$
Applying the chain rule, the gradient of $J$ with respect to
$\boldsymbol P_u^{\ast}$ is obtained as
\begin{equation}
\nabla_{\boldsymbol P_u^{\ast}} J
=
-
\frac{B_g}{\ln 2}
\sum_{l=1}^{L}
\frac{\eta_{l,u}}{\omega_l}
\frac{
\tilde{\boldsymbol M}_u^{H}\boldsymbol H_{l,u}^{H}
\big(
\boldsymbol H_{l,u}\tilde{\boldsymbol M}_u\boldsymbol P_u
\big)
}{
1-B
\big|
\boldsymbol H_{l,u}\tilde{\boldsymbol M}_u\boldsymbol P_u
\big|^2
}.
\end{equation}

The holographic beamforming matrix can be written as
$
\boldsymbol M_u
=
\sqrt{\beta}\,
\mathrm{diag}(\boldsymbol m_u)\,
\boldsymbol A_u,
$
where $\boldsymbol m_u=[m_{u,1},\dots,m_{u,N}]^{T}$ and 
$
\boldsymbol{A}_u(n, k) \triangleq e^{-\alpha\left\|\boldsymbol{r}_n^k\right\|} e^{-j \boldsymbol{k}_s \boldsymbol{r}_n^k}, \quad \boldsymbol{A}_u \in \mathbb{C}^{N \times K}.
$
$\boldsymbol A_u$ is independent of $a_{l,u,k}$.
Since
$
m_{u,n}
=
\sum_{l'=1}^{L}\sum_{k'=1}^{K}
a_{l',u,k'}\,
\phi_{l',u,k'}(n),
$
the derivative of $\tilde{\boldsymbol M}_u$ with respect to $a_{l,u,k}$ is
$
\frac{\partial \tilde{\boldsymbol M}_u}{\partial a_{l,u,k}}
=
\boldsymbol \Xi_u\frac{\partial \boldsymbol M_u}{\partial a_{l,u,k}}
=
\sqrt{\beta}\,
\boldsymbol \Xi_u\,
\mathrm{diag}(\boldsymbol \phi_{l,u,k})\,
\boldsymbol A_u,
$
where
$
\boldsymbol \phi_{l,u,k}(n)
=
\big(\operatorname{Re}\{\Psi_{l,u}^{\mathrm{intf}}(\boldsymbol r_n^k)\}+1\big)/2
$.
Using the chain rule, the gradient of $J$ with respect to
$a_{l,u,k}$ is given by $\nabla_{ a_{l,u,k}} J=$
\begin{equation}
-
\frac{2B_g}{\ln 2}
\sum_{q=1}^{L}
\frac{\eta_{q,u}}{\omega_q}
\frac{
\operatorname{Re}\left\{
\big(
\boldsymbol H_{q,u}\tilde{\boldsymbol M}_u\boldsymbol P_u
\big)^{\!*}
\boldsymbol H_{q,u}
\left(
\frac{\partial \tilde{\boldsymbol M}_u}{\partial a_{l,u,k}}
\right)
\boldsymbol P_u
\right\}
}{
1-\eta_{q,u}
\big|
\boldsymbol H_{q,u}\tilde{\boldsymbol M}_u\boldsymbol P_u
\big|^2
}.
\end{equation}

%
%

\end{document}